\documentclass{lmcs}
\pdfoutput=1

\usepackage{lastpage}
\lmcsdoi{15}{2}{20}
\lmcsheading{}{\pageref{LastPage}}{}{}%
{Mar.~30,~2018}{Jun.~11,~2019}{}

\keywords{Regular separability problem,  one counter automata, one counter nets, vector addition systems with states}

\usepackage{silence}
\usepackage{pdfsync}
\usepackage[utf8]{inputenc}
\usepackage{mathtools}
\usepackage{amsthm}
\usepackage{amsmath}
\usepackage{amssymb}
\usepackage{mathabx}
\usepackage{tikz}
\usepackage{comment}
\usepackage{xspace}
\usepackage{thm-restate}

\newcounter{mycnt}

\newcommand{\RR}{\mathcal{R}}

\newcommand{\finsubseteq}{\subseteq_\text{\tiny fin}}
\newcommand{\crossproduct}{synchronized product\xspace}
\newcommand{\Crossproduct}{Synchronized product\xspace}

\newtheorem{proposition}[mycnt]{Proposition}
\newtheorem{lemma}[mycnt]{Lemma}
\newtheorem{theorem}[mycnt]{Theorem}
\newtheorem{example}[mycnt]{Example}
\newtheorem{corollary}[mycnt]{Corollary}
\newtheorem{remark}[mycnt]{Remark}
\newtheorem{assumption}[mycnt]{Assumption}

\theoremstyle{thmC}
\newtheorem{theoremC}[mycnt]{Theorem}
\newtheorem{lemmaC}[mycnt]{Lemma}

\newcommand{\C}{\mathcal{C}}

\newcommand{\F}{\mathcal{F}}
\newcommand{\G}{\mathcal{G}}

\newcommand{\A}{\mathcal{A}}
\newcommand{\B}{\mathcal{B}}

\newcommand{\V}{\mathcal{V}}
\newcommand{\N}{\mathbb{N}}
\newcommand{\Z}{\mathbb{Z}}

\newcommand{\U}{\mathcal{U}}
\newcommand{\M}{\mathcal{M}}
\newcommand{\T}{\mathcal{T}}

\newcommand{\sol}[1]{\text{sol}(#1)} 

\newcommand{\Sigmaeps}{\Sigma_{\varepsilon}}
\newcommand{\deeps}[1]{#1_{|\Sigma}}

\newcommand{\trans}[1]{\stackrel{#1}{\longrightarrow}}
\newcommand{\myparagraph}[1]{\paragraph{\bf #1}}
\newcommand{\sepsep}{ }

\newcommand{\reach}{\textsc{Reach}}

\newcommand{\vass}{VASS\xspace}
\newcommand{\prevass}{integer-VASS\xspace}

\newcommand{\vasslangs}{\vass languages\xspace}
\newcommand{\set}[1]{\{#1\}}
\newcommand{\setof}[2]{\set{#1 \mid #2}} 
\newcommand{\oca}{OCA\xspace}
\newcommand{\ocn}{OCN\xspace}

\newcommand{\pspace}{{\sc PSpace}\xspace}
\newcommand{\expspace}{\textsc{ExpSpace}\xspace}

\newcommand{\pref}{\textsc{pref}}
\newcommand{\midd}{\textsc{mid}}
\newcommand{\suff}{\textsc{suff}}

\newcommand{\low}{\textsc{low}}
\newcommand{\high}{\textsc{high}}

\newcommand{\acc}{\textup{acc}}
\newcommand{\rej}{\textup{rej}}

\newcommand{\mycap}{\, \cap \,} 
\newcommand{\pim}[1]{\text{\sc Pi}(#1)} 

\newcommand{\modulo}[1]{\text{ mod } #1}
\newcommand{\synchrprod}{\times}
\newcommand{\crossprod}{\otimes}
\newcommand{\constrvass}[2]{#1(#2)}

\newcommand{\size}[1]{\text{size}(#1)} 

\newcommand{\good}{definite\xspace}
\newcommand{\pspacesolvable}{\pspace-enumerable\xspace}

\begin{document}

\title[Regular Separability of One Counter Automata]{Regular Separability of One Counter Automata\rsuper*}
\titlecomment{{\lsuper*}This is a thoroughly revised and extended version of~\cite{lics17}.}

\author[W.~Czerwi{\'n}ski]{Wojciech Czerwi{\'n}ski}
\author[S.~Lasota]{Sławomir Lasota}

\address{Wydział Matematyki, Informatyki i Mechaniki, University of Warsaw, Poland}
\email{\{wczerwin,sl\}@mimuw.edu.pl}
\thanks{The first author acknowledges partial support by the Polish NCN grant 2016/21/D/ST6/01376.}

\thanks{The second author acknowledges partial support
by the European Research Council (ERC) project Lipa under the EU Horizon 2020
research and innovation programme (grant agreement No. 683080).}

\begin{abstract}
The regular separability problem asks, for two given languages, if there exists
a regular language including one of them but disjoint from the other.
Our main result is decidability, and \pspace-completeness,
of the regular separability problem for languages of one counter automata without zero tests
(also known as one counter nets).
This contrasts with undecidability of the regularity problem for one counter nets, and with undecidability of
the regular separability problem for one counter automata, which is our second result.
\end{abstract}

\maketitle


\section{Introduction}
We focus on separability problems for languages of finite words.
We say that a language $K$ is \emph{separated from} another language $L$ by a language $S$,
if $K \subseteq S$ and $L \cap S = \emptyset$.
For two families of languages $\F$ and $\G$, the \emph{$\F${\sepsep}separability problem for $\G$}
asks, for two given languages $K, L \in \G$ over the same alphabet,
whether $K$ is separated from $L$ by some language from $\F$.

In this paper we mainly consider the separator class $\F$ of regular languages (thus using the term \emph{regular separability}).
As regular languages are closed under complement, $K$ is separated from $L$ by a regular language
if, and only if, $L$ is separated from $K$ by a regular language. Therefore we shortly say that $K$ \emph{and} $L$ are \emph{regular separable}.
As the class $\G$ we consider the languages of \emph{one counter automata} (NFA extended with a non-negative counter that can be
incremented, decremented and tested for zero), or its subclass --- the languages of \emph{one counter nets} (one counter automata without zero tests).

\myparagraph{Motivation and context}
Separability is a classical problem in formal languages.
It was investigated most extensively for $\G$ the class of regular languages, and for $\F$ a suitable subclass thereof.
Since regular languages are effectively closed under complement, the
$\F${\sepsep}separability problem is a generalization of the
\emph{$\F${\sepsep}characterization problem},
which asks whether a given regular language $L$ belongs to $\F$:
indeed, $L \in \F$
if, and only if, $L$ and its complement are separable by some language from $\F$.
Separability problems for regular languages were investigated since a long time, starting from
a generic connection established by Almeida~\cite{Almeida-pmd99} between profinite semigroup theory and separability.
Recently it attracted a lot of attention also outside algebraic community,
which resulted in establishing the decidability of $\F${\sepsep}separability for the family $\F$ of separators being, among others,
\begin{itemize}
\item the piecewise testable languages~\cite{DBLP:conf/icalp/CzerwinskiMM13,DBLP:conf/mfcs/PlaceRZ13}
\item the locally and locally threshold testable languages~\cite{DBLP:conf/fsttcs/PlaceRZ13},
\item the languages definable in first order logic~\cite{DBLP:journals/corr/PlaceZ14},
\item the languages of certain fixed levels of the first order hierarchy~\cite{DBLP:conf/icalp/PlaceZ14}.
\end{itemize}
The first result has been recently generalized to finite ranked trees~\cite{DBLP:conf/icalp/Goubault-Larrecq16}.

Separability of non-regular languages attracted little attention till now.
The reason for this may be twofold. First, for regular languages one can use standard algebraic tools, like syntactic monoids,
and indeed most of the results have been obtained using algebraic techniques. 
Second, the few known negative results on separability of non-regular languages are strongly discouraging.
To start off,
some strong intractability results have been known already since 70's, when Szymanski and Williams proved
that regular{\sepsep}separability of context-free languages is undecidable~\cite{DBLP:journals/siamcomp/SzymanskiW76}.
Later Hunt~\cite{DBLP:journals/jacm/Hunt82a} strengthened this result: he showed that $\F${\sepsep}separability of context-free languages
is undecidable for every class $\F$ containing all \emph{\good} languages,
i.e., finite Boolean combinations of languages of the form $w\Sigma^*$ for $w \in \Sigma^*$.
This is a very weak condition, hence the result of Hunt suggested
that nothing nontrivial can be done outside regular languages with respect to separability problems.
Furthermore, Kopczy\'{n}ski has recently shown that regular{\sepsep}separability is undecidable
even for languages of visibly pushdown automata~\cite{Kopczynski16},
thus strengthening the result by Szymanski and Williams once more.

On the positive side, piecewise testable{\sepsep}separability has been shown decidable
for context-free languages, languages of vector addition systems with states (\vasslangs), 
and some other classes of languages~\cite{DBLP:conf/fct/CzerwinskiMRZ15}.
This inspired us to start a quest for decidable cases beyond regular languages.

Once beyond regular languages, the regular{\sepsep}separability problem seems to be the most intriguing.
\vasslangs is a well-known class of languages, for which the decidability status of the regular{\sepsep}separability problem is unknown.
A few positive results related to this problem have been however obtained recently.
First, decidability of unary (and modular) separability of reachability sets\footnote{Note that these are sets of vectors, not words.}
of \vass was shown in~\cite{CCLP17};
the problem is actually equivalent to regular separability of commutative closures of \vasslangs.
Second, decidability of regular{\sepsep}separability of languages of Parikh automata was shown in~\cite{CCLP16}.
Parikh automata recognize exactly the same languages as \emph{integer-\vass} (a variant of \vass  where one allows
negative counter values~\cite{KR03,CFM11}), and therefore are a subclass of \vasslangs.
Finally, decidability of regular separability for \emph{coverability} languages of \vass follows from the generic result
of~\cite{CLMMKS18} for well-structured transition systems.

The open decidability status of regular separability of \emph{reachability} languages of \vass
is our main motivation in this paper.
A more general goal is understanding for which classes of languages
the regular separability problem is decidable.

\myparagraph{Our contribution}
We consider the regular{\sepsep}separability problem for languages of one counter automata (with zero test)
and its subclass, namely one counter nets (without zero test); the latter model is exactly \vass in dimension 1.
The two models we call shortly \oca and \ocn, respectively.
Our main result is decidability of the regular{\sepsep}separability problem for languages of one counter nets.
Moreover, we determine the exact complexity of the problem, namely \pspace-completeness.
For complexity estimations we assume a standard encoding of \oca (or \ocn) and their configurations;
in particular we assume binary encoding of integers appearing in the input.
\begin{theorem}\label{thm:pspace-comp}
Regular{\sepsep}separability of languages of \ocn is \pspace-complete.
\end{theorem}
%
%
%
%
%
Our approach to prove decidability is by \emph{regular over-approximation}: for every \ocn language $L$
there is a decreasing sequence of (computable) regular languages over-approximating $L$,
such that two \ocn languages are regular separable if, and only if, some pair of their approximants is disjoint.
Furthermore, the latter condition can be reduced to a kind of reachability property of the \crossproduct
of two \ocn, and effectively checked in \pspace by exploiting effective semi-linearity of the reachability set of
the \crossproduct.
Our \pspace lower bound builds on \pspace-hardness of bounded non-emptiness of \oca~\cite{FJ15}.

It is interesting to compare the regular{\sepsep}separability problem with the regularity problem,
which asks whether a given language is regular.
For every class $\G$ effectively closed under complement, regular{\sepsep}separability is a generalization of regularity,
as $L$ is regular if, and only if, $L$ and its complement $\bar{L}$ are regular{\sepsep}separable.
It turns out that regularity of \ocn languages can not be reduced to regular separability:
firstly because \ocn languages are not closed under complement, and secondly (and more importantly)
because the regularity problem is undecidable for \ocn languages~\cite{ValkVidal81,BensUndecidability}
while we prove regular{\sepsep}separability decidable.

As our second main contribution, we show that adding zero tests leads to undecidability, 
for any separator language class containing all \good languages.
In particular, regular languages are an example of such a class.
\begin{theorem}\label{thm:undecidability}
For every language class $\F$ containing all \good languages,
the $\F${\sepsep}separability problem  for languages of \oca is undecidable.
\end{theorem}
Our argument is inspired by the undecidability proof by Hunt~\cite{DBLP:journals/jacm/Hunt82a}: we show, roughly speaking, that
every decidable problem reduces in polynomial time to the separability problem for \oca.

\myparagraph{Organization}
In Section~\ref{sec:oca} we define the models of \oca and \ocn,
then Sections~\ref{sec:approx}--\ref{sec:pspace-hard} are devoted to the proof of Theorem~\ref{thm:pspace-comp},
and finally Section~\ref{sec:undecid} contains the proof of Theorem~\ref{thm:undecidability}.
The proof of Theorem~\ref{thm:pspace-comp} is factorized as follows:
in Section~\ref{sec:approx} we introduce the regular over-approximation of \ocn languages,
in Section~\ref{sec:decid} we we provide a decision procedure for testing the disjointness property of approximants,
as discussed above,
further in Section~\ref{sec:in-pspace} we provide a \pspace implementation of this procedure,
and in Section~\ref{sec:pspace-hard} we give a \pspace lower bound.
The last Section~\ref{sec:remarks} contains some concluding remarks, including the discussion of undecidability of the regularity problem for \ocn.




\section{One counter automata and nets}%
\label{sec:oca}

In order to fix notation we start by recalling finite automata with $\varepsilon$-transitions,
in a specifically chosen variant convenient when working with one counter automata.

A \emph{nondeterministic finite automaton} (NFA) $\A = (Q, q_0, q_f, T)$ over a finite alphabet $\Sigma$
consists of a finite set of control states $Q$, distinguished initial and
final states $q_0, q_f \in Q$ (for convenience we assume here, w.l.o.g., a single final state),
and a set of \emph{transitions} $T \subseteq Q \times \Sigmaeps \times Q$, where $\Sigmaeps = \Sigma\cup\set{\varepsilon}$.

For a word $v \in {(\Sigmaeps)}^*$, let $\deeps{v}$ be the word obtained by removing all occurrences of $\varepsilon$.
A run of $\A$ over a word $w\in\Sigma^*$ is a sequence of transitions of the form
\[(p_0, a_1, p_1), (p_1, a_2, p_2), \ldots, (p_{n-1}, a_n, p_n)\]
such that $\deeps{(a_1 \ldots a_n)} = w$.
The run is \emph{accepting} if $p_0 = q_0$ and $p_n = q_f$.
The language of $\A$, denoted $L(\A)$, is the set of all words $w$ over which $\A$ has an accepting run.
Languages of NFA are called \emph{regular}.

\myparagraph{One counter automata and nets}
In brief, a one counter automaton (\oca) is an NFA with a non-negative counter, where
we allow for arbitrary changes of the counter value in one step.

Formally, an \oca is a tuple $\A = (Q, \alpha_0, \alpha_f, T, T_{=0})$, where $Q$ are control states as above.
A \emph{configuration} $(q, n)\in Q\times\N$ of $\A$ consists of a control state and a non-negative counter value.
There are two distinguished configurations, the initial one $\alpha_0 = (q_0, n_0)$ and the final one $\alpha_f = (q_f, n_f)$.
The finite set $T \subseteq Q \times \Sigmaeps \times Q \times \Z$
contains \emph{transitions} of $\A$. A transition $(q, a, q', z)$ can be fired in a configuration $\alpha = (q, n)$ if $n+z \geq 0$, leading to a new configuration $\alpha' = (q', n+z)$.
We write $\alpha \trans{a} \alpha'$ if this is the case.
Finally, the set $T_{=0} \subseteq Q\times\Sigmaeps\times Q$ contains \emph{zero tests}. A zero test $(q, a, q')$ can be fired
in a configuration $\alpha = (q, n)$ only
if $n = 0$, leading to a new configuration $\alpha' = (q', n)$. Again, we write $\alpha \trans{a} \alpha'$ if this is the case.

A run of an \oca over a word $w\in \Sigma^*$ is a sequence of transitions and zero tests of the form
\[\alpha_0 \trans{a_1} \alpha_1 \trans{a_2} \ldots \trans{a_n} \alpha_n\]
such that $\deeps{(a_1 \ldots a_n)} = w$;
we briefly write $\alpha_0 \trans{w} \alpha_n$ if this is the case, and $\alpha_0 \trans{} \alpha_n$ if a word $w$ is irrelevant.
The run is \emph{accepting} if $\alpha_0$ is the initial configuration of $\A$, and $\alpha_n$ is the final one.
The language of $\A$, denoted $L(\A)$, is the set of all words $w$ over which $\A$ has an accepting run.

A one counter net (\ocn) is an \oca without zero tests, i.e., one with $T_{=0} = \emptyset$.
We drop the component $T_{= 0}$ and denote \ocn as $(Q, \alpha_0, \alpha_f, T)$.
In other words, an \ocn is exactly a vector addition system with states (\vass) in dimension 1~\cite{KM69,HP79}.

\begin{example}%
\label{ex:ocnsep}
Consider two \ocn languages over the alphabet $\set{a,b}$:
\[
K = \setof{a^n b^n}{n\in\N}
\qquad
L = \setof{a^n b^{n+1}}{n\in\N}.
\]
An example regular language separating $K$ from $L$ is
$R = \setof{a^n b^m}{n \equiv m \modulo{2}}$.
Indeed, $R$ includes $K$ and is disjoint with $L$.
On the other hand, $K$ and $L' = \setof{a^n b^m}{m > n}$ are not regular separable (which follows from
Corollary~\ref{cor:appr} below).
\end{example}

\myparagraph{Other modes of acceptance}
We briefly discuss other possible modes of acceptance of \oca.

First, consider a variant of \oca with a finite set of initial configurations, and a finite set of final ones.
This variant can be easily simulated by \oca as defined above.
Indeed, add two fresh states $q_0, q_f$, and fix the initial and final configurations $\alpha_0 = (q_0, 0)$ and
$\alpha_f = (q_f, 0)$. Moreover, add $\varepsilon$-transitions enabling to go from $\alpha_0$ to every of former initial configurations, and
symmetrically add $\varepsilon$-transitions enabling to go from every of former final configurations to $\alpha_f$.

The above simulation reveals that w.l.o.g.\ we can assume that the counter values $n_0$ and $n_f$
in the initial and final configurations are 0. This will be implicitly assumed in the rest of the paper.

Yet another possibility is accepting solely by control state: instead of a final configuration $\alpha_f = (q_f, n_f)$,
such an \oca would have solely a final control state $q_f$, and every run ending in a configuration $(q_f, n)$, for arbitrary
$n$, would be considered accepting.
Again, this variant is easily simulated by our model: it is enough to assume w.l.o.g.~that $q_f$ has no outgoing transitions nor zero tests,
add a transition $(q_f, \varepsilon, q_f, -1)$
decrementing the counter in the final state, and fix the final configuration as $(q_f, 0)$.

Finally, note that all the simulations discussed above work for \ocn as well.
In particular, in the sequel we may assume, w.l.o.g., that the counter values in initial and final configurations of \ocn are 0.



\section{Regular over-approximation of \ocn}%
\label{sec:approx}


For an \ocn $\A$ and $n > 0$, we are going to define an NFA $\A_n$ which we call \emph{$n$-approximation} of $\A$.
As the main result of this section we prove (as Corollary~\ref{cor:appr}) that two \ocn (and even \oca as
pointed out in Remark~\ref{rem:oca}) are regular separable if, and only if, their $n$-approximations are disjoint
for some $n>0$.

As long as the counter value is below $n$, the automaton $\A_n$ stores this value exactly
(we say then that $\A_n$ is in \emph{low} mode);
if the counter value exceeds $n$, the automaton $\A_n$ only stores the remainder of the counter value modulo $n$
(we say then that $\A_n$ is in \emph{high} mode).
Thus $\A_n$ can pass from low mode to high one; but $\A_n$ can also nondeterministically decide to pass the other way around, from high to low mode.

Let $Q$ be the state space of $\A$, and let $(q_0, 0)$ and $(q_f, 0)$ be its initial and final configurations.
As the state space of $\A_n$ we take the set
\[
Q_n = Q \times \{0, \ldots, n-1\} \times \{\low, \high\}.
\]
The initial and final state of $\A_n$ are $(q_0, 0, \low)$ and $(q_f, 0, \low)$, respectively.
Every transition $(q, a, q', z)$ of $\A$ induces a number of transitions of $\A_n$, as defined below
(for any $c$ satisfying $0 \leq c < n$):
\begin{align*}
& \big( (q, c, \low), a, (q', c{+}z, \low) \big) && \text{if } 0 \leq c{+}z < n \\
& \big( (q, c, \low), a, (q', (c{+}z) \modulo{n}, \high) \big) && \text{if } n \leq c{+}z  \\
& \big( (q, c, \high), a, (q', (c{+}z) \modulo{n}, \low) \big) && \text{if } c{+}z < 0  \\
& \big( (q, c, \high), a, (q', (c{+}z) \modulo{n}, \high) \big).
\end{align*}
Note that passing from high mode to low one is only possible if the counter value (modulo $n$) drops, after an update, strictly below 0;
in particular, this requires $z < 0$.
\begin{example}
Recall the languages $K$ and $L$ from Example~\ref{ex:ocnsep}, and consider an \ocn $\A$ recognizing $K$ that has two states $q_0$, $q_f$,
and three transitions:
\begin{align*}
& (q_0, a, q_0, +1) \\
& (q_0, \varepsilon, q_f, 0) \\
& (q_f, b, q_f, -1).
\end{align*}
The 2-approximating automaton $\A_2$ has 8 states $\set{q_0, q_f}\times\set{0, 1}\times\set{\low, \high}$.
In state $(q_0, 1, \low)$ on letter $a$, the automaton is forced to change the mode to $\high$; symmetrically,
in state $(q_f, 0, \high)$ on letter $b$, the automaton can change its mode back to $\low$:
\begin{align*}
\big( (q_0, 1, \low), a, (q_0, 0, \high) \big) \\
\big( (q_f, 0, \high), b, (q_f, 1, \low) \big).
\end{align*}
Otherwise, the mode is preserved by transitions; for instance, in high mode the automaton changes the state irrespectively of the input letter:
for every $q\in\set{q_0, q_f}$, $x\in\set{a, b}$ and $c\in\set{0,1}$, there is a transition
\begin{align*}
\big( (q, c, \high), x, (q, 1-c, \high) \big).
\end{align*}
The language recognized by $\A_2$ is \[\setof{a^n b^m}{(n = m < 2) \vee (n, m \geq 2 \wedge n \equiv m \modulo{2})}.\]
\end{example}
According to the definition above, the automaton $\A_n$ can oscillate between low and high mode arbitrarily many times.
Actually, as we argue below, it is enough to allow for at most one oscillation.
\begin{proposition}%
\label{prop:high}
For every run of $\A_n$ between two states 
in high mode,
there is a run over the same word between the same states which never exits the high mode.
\end{proposition}
\begin{proof}
Indeed, observe that if $\A_n$ has any of the following transitions
\begin{align*}
& \big( (q, m, \low), a, (q', m', \low) \big) \\
& \big( (q, m, \low), a, (q', m', \high) \big) \\
& \big( (q, m, \high), a, (q', m', \low) \big)
\end{align*}
then $\A_n$ necessarily has also the transition
\begin{align*}
\big( (q, m, \high), a, (q', m' \modulo{n}, \high) \big).
\end{align*}
Thus every run oscillating through high and low modes that starts and ends in high mode, can be simulated by a one that never exits high mode.
\end{proof}

A run of an \ocn $\A$ we call \emph{$n$-low}, if the counter value is strictly below $n$ in all configurations of the run.
Proposition~\ref{prop:appr-char} below characterizes the language of $\A_n$ in terms of runs of $\A$, and will be useful
 for proving the Approximation Lemma below.
Then Corollary~\ref{cor:appr-prop}, its direct consequence, summarizes some properties of approximation
useful in the sequel.

%
\begin{proposition}%
\label{prop:appr-char}
Let $\A = (Q, (q_0, 0), (q_f ,0), T)$ be an \ocn, and let $n > 0$.
Then $w \in L(\A_n)$ if, and only if,
\begin{enumerate}
\item[(a)] either $\A$ has an $n$-low accepting run over $w$,
\item[(b)] or $w$ factorizes into $w = w_\pref w_\midd w_\suff$, such that $\A$ has the following runs
\begin{align}
\begin{aligned} \label{eq:threeruns}
& (q_0, 0) \trans{w_\pref} (q, n+d) \\
& (q, c n+d) \trans{w_\midd} (q', c' n+d') \\
& (q', n+d') \trans{w_\suff} (q_f, 0),
\end{aligned}
\end{align}
for some states $q, q' \in Q$ and natural numbers $c, c' \geq 1$ and $d, d' \geq 0$.
\end{enumerate}
\end{proposition}

\begin{proof}
We start with the ``if'' direction. If there is an $n$-low run over $w$ in $\A$ then clearly $w \in L(\A_n)$.
Otherwise, suppose that $w = w_\pref w_\midd w_\suff$ and the words $w_\pref$, $w_\midd$ and $w_\suff$ admit the runs
as stated in~\eqref{eq:threeruns} above.
Then clearly $\A_n$ admit the following runs:
\begin{align*}
& (q_0, 0, \low)  \trans{w_\pref} (q, d \modulo{n}, \high) \\
& (q, d \modulo{n}, \high) \trans{w_\midd} (q', d' \modulo{n}, \high) \\
& (q', d' \modulo{n}, \high) \trans{w_\suff} (q_f, 0, \low)
\end{align*}
and thus $(q_0, 0) \trans{w} (q_f, 0)$ in $\A_n$ as required.

For the ``only if'' direction, suppose $w \in L(\A_n)$.
If $\A_n$ has a run over $w$ that never exits low mode, then clearly $\A$ has an $n$-low run over $w$.
Otherwise, consider any run of $\A_n$ over $w$. Distinguish the first and the last configuration in high mode along this run,
say $(q, d, \high)$ and $(q', d', \high)$.
The two configurations determine a factorization of the word $w$ into three parts $w = w_\pref w_\midd w_\suff$ such that
$\A_n$ admit the following runs:
\begin{align*}
& (q_0, 0, \low)  \trans{w_\pref} (q, d, \high) \\
& (q, d, \high) \trans{w_\midd} (q', d', \high) \\
& (q', d', \high) \trans{w_\suff} (q_f, 0, \low).
\end{align*}
The first and the last run imply the first and the last run in~\eqref{eq:threeruns}.
For the middle one, we may assume (w.l.o.g., by Proposition~\ref{prop:high}) that $\A_n$ never exits high mode,
which implies immediately existence of the middle run in~\eqref{eq:threeruns}.
\end{proof}
\begin{corollary}%
\label{cor:appr-prop}
Let $\A$ be an \ocn and let $m, n > 0$.
Then
\begin{enumerate}
  \item[(a)] $L(\A) \subseteq L(\A_n)$,
  \item[(b)] $L(\A_n) \subseteq L(\A_m)$ if $m \mid n$.
\end{enumerate}
\end{corollary}
\begin{proof}
The first inclusions follow easily by the characterization of Proposition~\ref{prop:appr-char}.
The second one is easily shown by definition of $n$-approximation.
\end{proof}
Now we state and prove the Approximation Lemma, which is the crucial property of approximation.
In the sequel we will strongly rely on direct consequences of this lemma, formulated as
Corollaries~\ref{cor:apprlemma} and~\ref{cor:appr} below.
\begin{lemma}[Approximation Lemma]%
\label{lem:appr}
For an \ocn $\A$, the following conditions are equivalent:
\begin{enumerate}
\item[(a)] $L(\A)$ is empty,
\item[(b)] $L(\A_n)$ is empty, for some $n > 0$.
\end{enumerate}
\end{lemma}

\begin{proof}
Clearly (b) implies (a), by Corollary~\ref{cor:appr-prop}(a). 
In order to prove that (a) implies (b), fix $\A = (Q, (q_0, 0), (q_f, 0), T)$ and
suppose that the languages $L(\A_n)$ are \emph{non-empty}
for \emph{all} $n > 0$; our aim is to show that $L(\A)$ is non-empty as well.

In the sequel we do not need the non-emptiness assumption for \emph{all} $n$;
it will be enough to use the assumption for some fixed $n$ computed as follows.
Let $|Q|$ be the number of states of $\A$ and $d_\A$ be the maximal absolute value of
integer constants appearing in transitions $T$ of $\A$.
Then let $K = 
|Q| \cdot d_\A$, and let $n = K!$ ($K!$ stands for $K$ factorial.) 

Let $w$ be a fixed word that belongs to $L(\A_n)$.
Our aim is to produce a word $w'$ that belongs to $L(\A)$,
by a pumping in the word $w$; the pumping
will allow to make a run of $\A_n$ into a correct run of $\A$.

As $w\in L(\A_n)$, by Proposition~\ref{prop:appr-char} we learn that $w$ satisfies
one of conditions (a), (b). If $w$ satisfies (a) then $w' = w\in L(\A)$ as required.
We thus concentrate, from now on, on the case when $w$ satisfies condition (b) in Proposition~\ref{prop:appr-char}.
Let's focus on the first (fixed from now on) run of $\A$ in~\eqref{eq:threeruns}, namely
\[
(q_0, 0) \trans{w_\pref} (q, n+d),
\]
for some prefix $w_\pref$ of $w$ and $d \geq 0$.
This run starts with the counter value $0$, and ends with the counter value at least $n$.
We are going to analyze closely the prefix of the run that ends immediately before the counter value exceeds $K$
for the first time; denote this prefix by $\rho$.
A configuration $(q, m)$ in $\rho$ we call \emph{latest} if the counter value stays strictly above
$m$ in all the following configurations in $\rho$. In other words, a latest configuration is the last one in $\rho$
where the counter value is at most $m$.
A crucial but easy observation is that the difference of counter values of two consecutive latest configurations is at most $d_\A$.
Therefore, as $K$ has been chosen large enough, $\rho$ must contain more than $|Q|$ latest configurations.
By the pigeonhole principle, there must be a state of $\A$, say $q$, that appears in at least
two 
latest configurations.
%
%
In consequence, for some infix $v$ of $w_\pref$, the \ocn $\A$ has a run over $v$ of the form
\[
(q, m) \trans{v} (q, m'), \quad \text{ for some } m < m' \leq m+K.
\]
As a consequence, the word $v$ can be repeated an arbitrary number of times,
preserving correctness of the run but increasing the final counter value.
Recall that the final counter value of $\rho$ is $n+d$, while we would like to achieve
$c n + d$ (for $c$ in Proposition~\ref{prop:appr-char}). Modify the word $w_\pref$ by adding
$(c-1) \cdot n / (m' -m)$ repetitions of the word $v$, thus obtaining a new word $w'_\pref$ such that $\A$ has a run
\begin{align} \label{eq:runpref}
(q_0, 0) \trans{w'_\pref} (q, cn + d).
\end{align}

In exactly the same way we modify the suffix $w_\suff$ of $w$, thus obtaining a word $w'_\suff$
over which the \ocn $\A$ has a run
\begin{align} \label{eq:runsuff}
(q', c' n + d') \trans{w'_\suff} (q_f, 0).
\end{align}
By concatenation we obtain a word $w' = w'_\pref w_\midd w'_\suff$ which is
accepted by $\A$, by composition of the run~\eqref{eq:runpref}, the middle run in~\eqref{eq:threeruns}, and the run~\eqref{eq:runsuff}.
Thus $L(\A)$ is non-empty, as required.
\end{proof}
As {\ocn}s are closed under products with finite automata and these products commute with $n$-approximations, we get: 
%
\begin{restatable}{corollary}{CorApprLemma}%
\label{cor:apprlemma}
For an \ocn $\A$ and a regular language $R$, the following conditions are equivalent:
\begin{enumerate}
\item[(a)] $L(\A)$ and $R$ are disjoint,
\item[(b)] $L(\A_n)$ and $R$ are disjoint, for some $n > 0$.
\end{enumerate}
\end{restatable}
\begin{proof}
Fix an \ocn $\A = (Q, (q_0, 0), (q_f, 0), T)$ and an NFA $\B = (P, p_0, p_f, U)$ recognizing the language $R$.
For convenience we assume here that $\A$ has an $\varepsilon$-transition of the form $(q, \varepsilon, q, 0)$ in every state $q\in Q$,
and $\B$ has a self-loop $\varepsilon$-transitions $(p, \varepsilon, p)$ in every state $p\in P$.
We will use the synchronized product $\A \synchrprod \B$ of \ocn $\A$ and NFA $\B$, which is the \ocn defined
by
\[
\A \synchrprod \B \ = \  \big(Q\times P, ((q_0, p_0), 0), ((q_f, p_f), 0), V\big),
\]
where a transition $\big((q, p), a, (q', p'), z\big) \in V$ if $(q, a, q', z) \in T$ and $(p, a, p') \in U$.
Observe that the synchronized product construction commutes with $n$-approximation: up to isomorphism of finite automata,
\begin{align} \label{eq:crossprodn}
{(\A \synchrprod \B)}_n  \ = \  \A_n \synchrprod \B.
\end{align}
Condition (a) in Corollary~\ref{cor:apprlemma} is equivalent to emptiness of the product $\A\synchrprod\B$ which,
by the Approximation Lemma applied to $A\synchrprod \B$, is equivalent to emptiness of the \emph{left} automaton in~\eqref{eq:crossprodn}, for some $n$.
Therefore condition (a) is also equivalent to emptiness of the \emph{right} automaton in~\eqref{eq:crossprodn}, for some $n$.
Finally, the latter condition is equivalent to condition (b).
%
%
%
\end{proof}

\begin{restatable}{corollary}{CorAppr}%
\label{cor:appr}
For two \ocn $\A$ and $\B$, the following conditions are equivalent:
\begin{enumerate}
\item[(a)] $L(\A)$ and $L(\B)$ are regular separable,
\item[(b)] $L(\A_n)$ and $L(\B)$ are disjoint, for some $n > 0$,
\item[(c)] $L(\A_n)$ and $L(\B_n)$ are disjoint, for some $n > 0$.
\end{enumerate}
\end{restatable}
\begin{proof}
In order to prove that (a) implies (b), suppose that a regular language $R$ separates $L(\B)$ from $L(\A)$, i.e., $R$ includes $L(\B)$ and
is disjoint from $L(\A)$. By Corollary~\ref{cor:apprlemma} we learn that for some $n > 0$, $R$ and $\A_n$ are disjoint. Thus necessarily
$L(\B)$ and $L(\A_n)$ are disjoint too.

To show that (b) implies (c) use Corollary~\ref{cor:apprlemma} for \ocn $\B$ and regular language $L(\A_n)$.
We get that there exists $m > 0$ such that $L(\B_m)$ and $L(\A_n)$ are disjoint. Then
using Corollary~\ref{cor:appr-prop}(b) we have that $L(\A_{nm})$ and $L(\B_{nm})$ are disjoint as well. 

Finally, (c) easily implies (a), as any of the regular languages $L(\A_n)$, $L(\B_n)$ can serve as a separator
(Corollary~\ref{cor:appr-prop}(a) is used here). 
\end{proof}

At this stage we do not have yet any effective bound on the minimal $n$ satisfying condition (b) or (c) in
 Corollary~\ref{cor:appr}.
Our decision procedure for \ocn, to be presented in the next section,
will test condition (b). 
A bound on $n$ in Corollary~\ref{cor:appr}(b)--(c) can be extracted from the decision procedure 
(as discussed in~Section~\ref{sec:in-pspace});
however, this bound does not lead directly to the optimal \pspace complexity.

\begin{remark}%
\label{rem:oca} \rm
Interestingly, exactly the same notion of approximation can be defined for \oca as well.
Even if Propositions~\ref{prop:high} and~\ref{prop:appr-char} are no more valid for \oca,
all other facts proved in this section still hold for this more general model,
in particular the Approximation Lemma and Corollaries~\ref{cor:apprlemma} and~\ref{cor:appr}.
Confronting this with undecidability of regular separability for \oca (which we prove in Section~\ref{sec:undecid}) leads to a conclusion that
the characterizations of  Corollary~\ref{cor:appr} are not effectively testable in case of \oca, while they are in case of \ocn.
\end{remark}



\section{Decision procedure}\label{sec:decid}

Our proof of \pspace-membership of the regular separability problem for \ocn splits into two parts.
In this section we do the first step:
we reduce the (non-)separability problem of two \ocn $\A$ and $\B$ 
to a kind of reachability property in the \crossproduct of $\A$ and $\B$,
and then build the decision procedure relying on semi-linearity of the corresponding reachability relation.
As the second (more technical) step, in Section~\ref{sec:in-pspace}
we concentrate on implementing the decision procedure in \pspace:
we encode the reachability property using (multiple) systems of linear Diophantine equations
which will be all enumerable (and hence solvable) in polynomial space.

\myparagraph{Semi-linear sets}
For a set $P \subseteq \Z^l$ of vectors, let $P^*\subseteq \Z^l$ contain all vectors that can be obtained as a finite sum, possibly the empty one,
and possibly with repetitions, of vectors from $P$.
In other words, $P^*$ is the set of \emph{non-negative} linear combinations of vectors from $P$.
\emph{Linear sets} are sets of the form $L = \set{b} + P^*$, where $b\in \Z^l$, $P$ is a finite subset of $\Z^l$,
and addition $+$ is understood element-wise.
Thus $L$ contains sums of the vector $b$ and a vector from $P^*$.
The vector $b$ is called \emph{base}, and vectors in $P$ \emph{periods}; we write
shortly $b + P^*$.
Finite unions of linear sets are called \emph{semi-linear}.
We use sometimes a special case of semi-linear sets of the form $B + P^*$,
for finite sets $B, P$.

We will often consider (semi-)linear sets with some coordinates non-negative, 
e.g., $b + P^* \subseteq \N^2\times\Z$. Note that in this case we necessarily have
$b\in\N^2 \times \Z$ and $P \subseteq \N^2\times\Z$.

\myparagraph{Vector addition systems with states}
We start by recalling the notion of \emph{integer} vector addition systems with states (\prevass).
For $d > 0$, a $d$-dimensional \prevass $\V = (Q, T)$, or $d$-\prevass, consists of a finite set $Q$ of control states, and a finite
set of transitions $T \subseteq Q\times \Z^d\times Q$.
A configuration of $\V$ is a pair $(q, v) \in Q\times\Z^d$ consisting of a state and an integer vector.
Note that we thus allow, in general, negative values in configuration (this makes a difference between \prevass and \vass);
however later we will typically
impose non-negativeness constraints on a selected subset of coordinates.
A $d$-\prevass $\V$ determines a step relation between configurations: there is a step from $(q, v)$ to $(q', v')$ if
$T$ contains a transition $(q, z, q')$ such that $v' = v + z$.
We write $(q, v) \trans{} (q', v')$ if there is a sequence of steps leading from $(q, v)$ to $(q', v')$, and say that
$(q', v')$ is \emph{reachable} from $(q, v)$ in $\V$.

\myparagraph{\Crossproduct}
For convenience we assume that every \ocn has an $\varepsilon$-transition of the form
$(q, \varepsilon, q, 0)$ in every control state $q$.
We will use a \crossproduct operation over one counter nets.
For two \ocn $\A = (Q, \alpha_0, \alpha_f, T)$ an $\B = (P, \beta_0, \beta_f, U)$, their \emph{\crossproduct}
$\A \crossprod \B$ is a 2-\prevass
whose
states are pairs of states $Q\times P$ of $\A$ and $\B$, respectively, and whose transitions contain all triples
\[
\big( (q, p), (z, v), (q', p') \big)
\]
such that there exists $a\in \Sigmaeps$ with $(q, a, q', z) \in T$ and $(p, a, p', v) \in U$.
%
%
Note that $\A \crossprod \B$ is unlabeled --- the alphabet letters are only used to synchronize $\A$ and $\B$ ---
and allows, contrarily to $\A$ and $\B$, for negative values on both coordinates.
Moreover note that there is no distinguished initial or final configuration in an \prevass.

We will later need to impose a selective non-negativeness constraint on values of configurations.
For a $d$-\prevass $\V$ and a sequence $C_1, \ldots, C_d$, where $C_i = \N$ or $C_i = \Z$ for each $i$, by
$\constrvass{\V}{C_1, \ldots, C_d}$ we mean the transition system of $\V$ truncated to the subset
$Q\times C_1 \times \cdots \times C_d \subseteq Q\times\Z^d$ of configurations.
For instance, $\constrvass{(\A\crossprod\B)}{\N, \N}$ differs from $\A \crossprod \B$ by imposing the non-negativeness constraint
on both coordinates, and is thus a 2-\vass.
On the other hand, in $\constrvass{(\A\crossprod\B)}{\Z, \N}$ the counter of $\A$ can get arbitrary integer values while the counter of $\B$
is restricted to be non-negative.

\myparagraph{Disjointness assumption}
Fix, for this and the next section, two input \ocn
\[\A = (Q, (q_0, 0), (q_f, 0), T) \quad \text{ and } \quad \B = (P, (p_0, 0), (p_f, 0), U),\]
and let $\V = \A\crossprod\B$ be their \crossproduct.
If the intersection of $L(\A)$ and $L(\B)$ is non-empty, the answer to the separability question is obviously negative.
We may thus consider only input \ocn $\A$ and $\B$ with $L(\A)$ and $L(\B)$ are disjoint.
This is eligible as the disjointness can be effectively checked in \pspace.
Indeed, the intersection of $L(\A)$ and $L(\B)$ is nonempty if, and only if,
\[
\big( (q_0, p_0), 0, 0 \big) \trans{} \big( (q_f, p_f), 0, 0 \big)
\]
in the 2-\vass $\constrvass{\V}{\N, \N}$, which can be checked in \pspace
by the result of~\cite{DBLP:conf/lics/BlondinFGHM15}.

\begin{assumption}
For the decision procedure we assume, w.l.o.g.,~that $L(\A)\cap L(\B) = \emptyset$.
\end{assumption}

%
%

\myparagraph{Reduction to a reachability property of $ \A\crossprod\B$}
Recall Corollary~\ref{cor:appr}(b) which characterizes regular non-separability by non-emptiness of the intersection of 
$L(\A_n)$ and $L(\B)$, for all $n > 0$, which, roughly speaking, is equivalent to a reachability property
in the \crossproduct of NFA $\A_n$ and the \ocn $\B$, for all $n > 0$.
We are going now to internalize the quantification over all $n$, by transferring the reachability property to the
\crossproduct $\V= \A\crossprod\B$ of the two \ocn $\A$ and $\B$.

For convenience we introduce the following terminology. For $n > 0$ we say that $\V$ \emph{admits $n$-reachability} (or
$n$-reachability holds in $\V$) if
there are $q, q' \in Q$, $p, p'\in P$, $m, m' \geq n$, $l, l' \geq 0$ and $x \in \Z$ such that
$n | x$
and the following conditions hold:
\begin{align*}
  \pref_{q}^{p}(m, l): \quad &
  \big( (q_0, p_0), 0, 0 \big) \trans{} \big( (q, p), m, l \big) \text{ in } \constrvass{\V}{\N, \N}, \\
  \midd_{q q'}^{p p'}(m, l, m'+x, l'): \quad &
  \big( (q, p), m, l \big) \ \trans{} \big( (q', p'), m' + x, l' \big) \text{ in } \constrvass{\V}{\Z, \N}, \\
  \suff_{q'}^{p'}(m', l'): \quad &
  \big( (q', p'), m', l' \big) \trans{} \big( (q_f, p_f), 0, 0 \big) \text{ in } \constrvass{\V}{\N, \N}.
\end{align*}
The $n$-reachability in $\V$ differs  in three respects from ordinary reachability
$\big( (q_0, p_0), 0, 0 \big) \trans{} \big( (q_f, p_f), 0, 0 \big)$
in $\constrvass{\V}{\N, \N}$.
First, we require two intermediate values of the counter in $\A$, namely $m, m'$, to be at least $n$.
Second, in the middle part we allow the counter of $\A$ to be negative.
Finally, we allow for a mismatch $x\in \Z$ between the the middle and the final part,
as long as $x$ is divisible by $n$.
Thus $n$-reachability does \emph{not} imply non-emptiness $(q_0, 0) \trans{} (q_f, 0)$ of $\A$.
On the other hand, $n$-reachability \emph{does} imply non-emptiness $(p_0, 0) \trans{} (p_f, 0)$ of $\B$.

\begin{proposition}%
\label{prop:three-parts}
$\A$ and $\B$ are not regular separable if, and only if, $\V$ admits $n$-reachability for all $n > 0$.
\end{proposition}

\begin{proof}
Using the characterization of Corollary~\ref{cor:appr}(b), it suffices to show that for every $n > 0$, 
 $L(\A_n) \mycap L(\B) \neq \emptyset$ if, and only if, $\V$ admits $n$-reachability.
Fix $n > 0$ in the sequel.

For the ``only if'' direction, let $w \in L(\A_n) \mycap L(\B)$.
As $w \in L(\A_n)$, we may apply Proposition~\ref{prop:appr-char}.
Note that the condition (a) of Proposition~\ref{prop:appr-char} surely does not hold, as $w \notin L(\A)$
due to the disjointness assumption;
therefore condition (b) must hold
for some states $q, q' \in Q$ and natural numbers $c, c' \geq 1$ and $d, d' \geq 0$.
Put $m := n + d$, $m' := n + d'$ and
$x := (c' - c + 1)n$ (recall that $m' + x$ may be negative).
As $w\in L(\B)$, the corresponding states $p, p'$ and counter values $l, l'$ can be taken from the corresponding two positions
in an accepting run of $\B$ over $w$.
The chosen states $q, q', p, p'$ and integer values $m, m', l, l', x$ prove $n$-reachability in $\V$, as required.

For the ``if'' direction suppose that $\V$ admits $n$-reachability, and let $w_\pref$, $w_\midd$ and $w_\suff$ be some words
witnessing the three conditions of $n$-reachability:
\begin{align*}
  & \big( (q_0, p_0), 0, 0 \big) \trans{w_\pref} \big( (q, p), m, l \big) \text{ in } \constrvass{\V}{\N, \N}, \\
  & \big( (q, p), m, l \big) \ \trans{w_\midd} \big( (q', p'), m' + x, l' \big) \text{ in } \constrvass{\V}{\Z, \N}, \\
  & \big( (q', p'), m', l' \big) \trans{w_\suff} \big( (q_f, p_f), 0, 0 \big) \text{ in } \constrvass{\V}{\N, \N},
\end{align*}
for some $q, q' \in Q$, $p, p'\in P$, $m, m' \geq n$, $l, l' \geq 0$ and $x \in \Z$.
In particular, this implies
\begin{align}  \label{eq:wgore}
(q, m+(c-1)n) \trans{w_\midd} (q', m'+x+(c-1)n) \text{ in } \A
\end{align}
for $c\geq 1$ large enough.
This also implies that the word $w = w_\pref w_\midd w_\suff$ belongs to $L(\B)$.
We will prove that $w$ also belongs to $L(\A_n)$, by demonstrating
that the factorization $w = w_\pref w_\midd w_\suff$ satisfies the condition (b) in  Proposition~\ref{prop:appr-char}.
Indeed, for $d := m - n$, $d' := m' - n$, we obtain then runs over $m_\pref$ and $m_\suff$ as required in  (b)
in Proposition~\ref{prop:appr-char}.
In order to get a  run over $w_\midd$, we take $c\geq 1$ large enough so that~\eqref{eq:wgore} holds;
for $c' := c + x/n$,~\eqref{eq:wgore} rewrites, as required, to
$
(q, cn+d) \trans{w_\midd} (q', c'n + d') \text{ in } \A.
$
\end{proof}

\myparagraph{Witnesses and effective semi-linearity}

Building on Proposition~\ref{prop:three-parts}, we are going to design a decision procedure to check
whether $\V$ admits $n$-reachability for all $n>0$.
To this end we use the three conditions of $n$-reachability as subsets of pairs resp.~quadruples of integers,
\begin{align} \label{eq:relations}
\pref_{q}^{p}, \ \suff_{q}^{p} \subseteq \N^2, \qquad
\midd_{q q'}^{p p'} \subseteq \N^2 \times\Z \times \N,
\end{align}
%
%
and define the set $\RR \subseteq \N^2 \times \Z$ by the following formula with three variables $(m, m', x)$:
\newcommand{\defR}{& \exists q, q' \in Q,  \, p, p'\in P, \,  l, l' \in \N,  
\ \pref_q^p(m, l)  \land  \midd_{q q'}^{p p'}(m, l, m' {+} x, l')  \land     \suff_{q'}^{p'}(m', l'). 
}
\begin{align} \begin{aligned}
    \label{eq:defR}
    \defR
\end{aligned} \end{align}
Then $n$-reachability is equivalent to saying that some $(m, m', x)\in \RR$  satisfies
\begin{align} \label{eq:witness}
m, m' \geq n \quad \text{ and } \quad n | x.
\end{align}
Any triple $(m, m', x) \in \RR$ satisfying the condition~\eqref{eq:witness} we call \emph{$n$-witness} in the sequel.
In this terminology, our algorithm is to decide whether $\RR$ contains $n$-witnesses for all $n>0$.

\begin{proposition}%
\label{prop:semi-lin}
The set $\RR$ is effectively semi-linear, i.e., a union of linear sets
\begin{align} \label{eq:R}
\RR = L_1 \cup \ldots \cup L_k, 
\end{align}
where $L_i = b_i + {P_i}^*$ for effectively computable
bases $b_i \in \N^2\times\Z$ and periods $P_i\finsubseteq \N^2 \times Z$ ($i = 1, \ldots, k$).
\end{proposition}
\begin{proof}
All the sets in~\eqref{eq:relations} are effectively semi-linear.
Indeed, $\pref_q^p$ is essentially the reachability set of a 2-\vass, and thus effectively
semi-linear~\cite{DBLP:conf/lics/BlondinFGHM15,HP79}, and likewise for $\suff_q^p$.
Moreover, effective semi-linearity of $\midd_{q q'}^{p p'}$ can be derived directly from
Parikh's theorem; specifically, it follows from Lemma 3.4 in~\cite{Georg}
(in Section~\ref{sec:in-pspace} an explicit proof is provided, giving additionally an exponential bound on
representation size of the semi-linear set).
In consequence, as semi-linear sets are effectively closed under boolean combinations and projections,
and the ternary relation of addition is semi-linear, the set $\RR$ is effectively semi-linear too.
\end{proof}

The next lemma allows us to consider each of the linear sets separately:
\begin{lemma}%
\label{lem:unionwitness}
If a finite union $X_1 \cup \ldots \cup X_k \subseteq \N^2\times\Z$
contains $n$-witnesses for all $n>0$,
then some of $X_1, \ldots, X_k$ also does.
\end{lemma}
\begin{proof}
We use a monotonicity property: if $n' | n$ then
every $n$-witness is automatically also $n'$-witness.
Consider a sequence of $(n!)$-witnesses,  
for $n>0$, contained in $X$. One of the sets $X_1, \ldots, X_k$ necessarily
contains infinitely many of them.
By monotonicity, this set contains $(n!)$-witnesses for all $n>0$, and hence $n$-witnesses for all $n>0$.
\end{proof}

\myparagraph{Decision procedure}

Relying on Proposition~\ref{prop:semi-lin} and Lemma~\ref{lem:unionwitness}, our procedure
enumerates the linear sets~\eqref{eq:R} and chooses one of them.
It thus remains to solve the core problem: given $b \in \N^2\times\Z$
and $P \finsubseteq \N^2 \times \Z$, decide whether $L = b + P^*$
contains $n$-witnesses for all $n>0$.
For such sets $L$, the condition we are to check  boils down to two separate sub-conditions:
\begin{lemma}%
\label{lem:separately}
$L = b+P^*$ contains $n$-witnesses for all $n>0$ if, and only if,
 \begin{enumerate}
\item[(a)] for every $n$, there is $(m, m', x) \in L$ with $m, m' \geq n$; and
\item[(b)] for every $n$, there is $(m, m', x)\in L$ with $n | x$.
\end{enumerate}
\end{lemma}
\begin{proof}
Put $b = (b_1, b_2, b_3)$.
Indeed, if $(b_1, b_2, b_3) + (k_1, k_2, k_3) \in L$  for $b_1 + k_1, b_2 + k_2 \geq n$,
and $(b_1, b_2, b_3) + (m_1, m_2, m_3) \in L$ for $n | (b_3 + m_3)$, then
$(b_1, b_2, b_3) + n(k_1, k_2, k_3) + (m_1, m_2, m_3) \in L$ is an $n$-witness.
Hence conditions (a) and (b) imply that $L$ contains $n$-witnesses for all $n>0$.
The opposite direction is obvious.
\end{proof}

Since all vectors in $P$ are non-negative on the first two coordinates,
condition (a) in Lemma~\ref{lem:separately} is easy for algorithmic verification: enumerate vectors in $P$
while checking
whether some vector  is positive on first coordinate, and some (possibly different) vector
is positive on second coordinate.

As the last bit of our decision procedure,
it remains to check condition (b) in Lemma~\ref{lem:separately}.
Writing $b_3$, resp.~$P_3$, for the projection of $b$, resp.~$P$, on the third coordinate, we need to check
whether the set $b_3 +  {P_3}^* \subseteq \Z$ contains (possibly negative) multiplicities of all $n>0$.
We build on: 
\begin{restatable}{lemma}{PropLast}%
\label{lem:last}
The set $b_3 + {P_3}^*$ contains multiplicities of all $n>0$ if, and only if,
$b_3$ is a linear combination of $P_3$, i.e.,
\begin{align}%
\label{eq:proplast}
b_3 = a_1 p_1 + \cdots + a_k p_k,
\end{align}
for $a_1, \ldots, a_k \in \Z$ and $p_1, \ldots, p_k \in P_3$.
\end{restatable}
\begin{proof}
For the ``only if'' direction,
suppose that $b_3 + {P_3}^*$ contains multiplicities of all positive numbers.
If $b_3 = 0$ then it is the empty linear combination of $P_3$; suppose therefore that $b_3 \neq 0$.
Note that this implies in particular that $P_3$ is forcedly nonempty.
Fix an arbitrary $n \in P_3$. Suppose $n>0$ (if $n<0$ take $-n$ instead of $n$).
By the assumption,
$b_3 + p \equiv 0 \modulo{n}$ for some  $p\in {P_3}^*$, i.e., 
\[
b_3 \equiv - p \modulo{n}.  
\]
Then $b_3  = - p + a n$ for some $a\in\Z$,
hence 
a linear combination of $P_3$ as required.

For the ``if'' direction, suppose $b_3$ is a linear combination of $P_3$ as in~\eqref{eq:proplast},
and let $n>0$.
It is possible to decrease the numbers $a_1, \ldots, a_k$ by multiplicities of $n$ so that they become non-positive.
Thus we have
\[
b_3 \equiv (a_1 - c_1 n) p_1 + \cdots + (a_k - c_k n) p_k \modulo{n},
\]
for $a_1 - c_1 n \leq 0, \ldots, a_k - c_k n \leq 0$, i.e., $b_3 \equiv - p \modulo{n}$ for
$p = (c_1 n - a_1) p_1 + \cdots + (c_k n - a_k) p_k \in {P_3}^*$.
In consequence $b_3 + p \equiv 0 \modulo{n}$, 
as required.
\end{proof}

Thus we only need to check whether $b_3$ is a linear combination of $P_3$.
By the Chinese remainder theorem, this is equivalent to
$b_3$ being a multiplicity of the greatest common divisor 
of all numbers in $P_3$.
Thus our decision procedure enumerates the set $P$, computes the greatest common divisor $g$
of projections $p_3$ on the third coordinate of all vectors $p\in P$, and finally checks whether $g | b_3$.
This completes the description of the decision procedure.



\section{\pspace upper bound}\label{sec:in-pspace}

In this section we prove the \pspace upper bound of Theorem~\ref{thm:pspace-comp}.
All the \pspace complexity statements below are understood with respect to the size of the two input \ocn,
under binary encoding of integers.  

As before, fix two \ocn $\A$ and $\B$ with disjoint languages.
The decision procedure from Section~\ref{sec:decid} \emph{enumerates} all linear sets
$L = b + P^*$ appearing in an effectively computed representation of the semi-linear set $\RR$.
For obtaining the upper bound we need to provide suitable estimations on representation size of these linear sets.
To this aim we introduce the concept of \emph{\pspacesolvable sets}, whose semi-linear representation
can be effectively enumerated in polynomial space.

\myparagraph{\pspacesolvable sets}
%
For a finite set of vectors $P \finsubseteq \Z^l$,
we say that an algorithm \emph{enumerates} $P$ if it computes consecutive elements of a sequence
$p_1, \ldots, p_m$, possibly with repetitions, such that $P = \set{p_1, \ldots, p_m}$; in other words, every element of
$P$ appears at least once in the sequence, but no other element does.
An algorithm enumerates a linear set $L = b + P^*\subseteq \Z^l$
if it first computes $b$ and then enumerates $P$.
If there is a polynomial space algorithm which enumerates $L = b + P^*$, the set $L$ is called \emph{\pspacesolvable}.
A semi-linear set $S$ we call \pspacesolvable (slightly abusing the notation) if for some sequence of linear sets
$L_1, \ldots, L_k$ such that
\[
S = L_1 \cup \ldots \cup L_k,
\]
there is a polynomial space algorithm that
first enumerates $L_1$, then enumerates $L_2$, and so on, and finally enumerates $L_k$.
In particular, this means that for some polynomial bound $N$,
every base and every period can be stored using at most $N$ bits.

\myparagraph{\pspace upper bound}

Propositions~\ref{prop:effsolv} and~\ref{prop:effsolv2} below state that all the sets appearing
in~\eqref{eq:relations}
are \pspacesolvable.
Their direct consequence, Proposition~\ref{prop:Reffsolv}, says the same about the set $\RR$,
and forms the cornerstone of \pspace decision procedure.
\begin{proposition}%
\label{prop:effsolv}
For every $q \in Q$ and $p \in P$, the semi-linear sets
$\pref_{q}^{p}$ and $\suff_{q}^{p}$ are \pspacesolvable.
\end{proposition}
\begin{proposition}%
\label{prop:effsolv2}
For every $q, q' \in Q$ and $p, p' \in P$, the semi-linear set
$\midd_{q q'}^{p p'}$ is \pspacesolvable.
\end{proposition}
\begin{restatable}{proposition}{PropR}%
\label{prop:Reffsolv}
The set $\RR$ is \pspacesolvable.
\end{restatable}
%
%
\noindent
Before proving Propositions~\ref{prop:effsolv}--\ref{prop:Reffsolv}, we notice that
Proposition~\ref{prop:Reffsolv}
allows us to implement in polynomial space the enumeration of the semi-linear set $\RR$ which underlies
the decision procedure presented in Section~\ref{sec:decid}.
This yields the upper bound of Theorem~\ref{thm:pspace-comp}.

\myparagraph{Bound on the size of separator}

A further consequence of Proposition~\ref{prop:Reffsolv}
is a bound on the minimal value of $n$ in Corollary~\ref{cor:appr}(b), 
which naturally leads to a complexity upper bound.
Indeed, due to Proposition~\ref{prop:Reffsolv} such a bound can be extracted from the proofs of Lemmas~\ref{lem:unionwitness}--\ref{lem:last} but,
importantly, it is only \emph{doubly exponential}
(hence, exhaustive checking if $L(\A_n) \cap L(\B) \neq \emptyset$ for all $n$ so bounded would only yield
an \expspace algorithm).
First, the proof of Lemma~\ref{lem:last} reveals that if the set $b_3 + {P_3}^*$ does not contain multiplicities of all
$n>0$, and $P_3$ is nonempty, then it does not contain multiplicities of some period $n\in P_3$, hence $n$ is bounded exponentially
(one can obtain an even better bound for $n$, namely the greatest common divisor of all periods from $P_3$, which still is only exponentially bounded in general).
Thus by Lemma~\ref{lem:separately} we get an exponential bound on the smallest $n$ such that a linear set $L$ does not contain an $n$-witness.
Then we lift the bound to semi-linear sets (cf.~Lemma~\ref{lem:unionwitness}) but at the price of increasing it
to doubly exponential:
indeed, if every component linear set $L_i$ does not contain an $n_i$-witness for some $n_i>0$,
the semi-linear set $\RR = L_1\cup \ldots\cup L_k$ does not contain
an $n$-witness, for $n$ the least common multiplicity of all $n_i$ and hence doubly exponential in general.
We have thus bounded the smallest $n$ such that $n$-reachability does not hold in $\V$, and consequently
(cf.~the proof of Proposition~\ref{prop:three-parts}) $n$ in Corollary~\ref{cor:appr}(b). 
We do not know if this bound can be improved to single exponential (which would immediately make the exhaustive check a \pspace algorithm).
The question seems to be quite challenging, as the proofs of Propositions~\ref{prop:effsolv}--\ref{prop:Reffsolv}
rely on  a combination of several nontrivial results~\cite{Hofman16,DBLP:conf/lics/BlondinFGHM15,Kopczynski10}.

The rest of Section~\ref{sec:in-pspace} is devoted to the proofs of Propositions~\ref{prop:effsolv}--\ref{prop:Reffsolv}.

\subsection{Proof of Proposition~\ref{prop:effsolv}}%
\label{sec:proofeffsolv}

We concentrate on showing that the sets $\pref_{q}^{p}$ are \pspacesolvable.
(The sets $\suff_q^p$ can be dealt with in exactly the same way as $\pref_q^p$, but with $\V$ replaced by the
reverse of $\V$.)
In the sequel fix states $q, p$ of $\A$ and $\B$, respectively.
The set $\pref_q^p$ is nothing but the reachability set of a 2-\vass $\constrvass{\V}{\N,\N}$ in control state $(q, p)$,
from the initial configuration
$((q_0, p_0), 0, 0)$.
We build on a result of~\cite{DBLP:conf/lics/BlondinFGHM15} which describes the reachability set in terms of
sets reachable via a finite set of \emph{linear path schemes}, a notion that we are going to recall now.

Let $T$ be set of transitions of $\V$.
A linear path scheme is a regular expression over $T$ of the form:
\begin{align} \label{eq:lps}
E  = \alpha_0 \beta_1^* \alpha_1 \ldots \beta_k^* \alpha_k,
\end{align}
where $\alpha_i, \beta_i \in T^*$. The sequences $\beta_1, \ldots, \beta_k$ are called \emph{loops} of $E$.
By \emph{length} of $E$ we mean the sum of lengths of all $\alpha_i$ and $\beta_i$.
Let $\reach_E$ (the reachability set via $E$) contain all pairs $(n, m)\in\N^2$ such that
$((q_0, p_0), 0, 0) \trans{} ((q, p), n, m)$ in $\constrvass{\V}{\N, \N}$ via a sequence of transitions that belongs to $E$.

Here is Thm.~1 in~\cite{DBLP:conf/lics/BlondinFGHM15}, translated to our terminology:
\begin{lemmaC}[\cite{DBLP:conf/lics/BlondinFGHM15}]%
\label{lem:blondyn}
There are computable bounds $N_1$, $N_2$, where $N_1$ is exponential and $N_2$ is polynomial in the size of $\V$,
such that $\pref_q^p$ is the union of sets $\reach_E$, for linear path schemes $E$ of length at most $N_1$,
with at most $N_2$ loops.
\end{lemmaC}

Our decision procedure, instead of dealing explicitly with a linear path scheme $E$,
will rather use $4k+2$ pairs of integers characterizing $E$, as described above.
Let $a_i \in \Z^2$, for $i = 0, \ldots, k$, be the total effect of executing the sequence $\alpha_i$, i.e.,
$a_i$ is the sum of effects of consecutive transitions in $\alpha_i$,
and likewise $b_i$ for the sequence $\beta_i$, for $i = 1, \ldots, k$.
Moreover, let $c_i \in \N^2$, for $i = 0, \ldots, k$ be the (point-wise) minimal nonnegative values of counters that allow to execute the sequence $\alpha_i$ (in $\constrvass{\V}{\N, \N}$),
and likewise $d_i$ for the sequence $\beta_i$, for $i = 1, \ldots, k$.
The $4k+2$ pairs of numbers, namely $a_i, c_i$ (for $i = 0, \ldots, k$) and $b_i, d_i$ (for $i = 1, \ldots, k$),
we jointly call the \emph{profile} of the linear path scheme $E$.
\begin{lemma}%
\label{lem:checknumbers}
Given pairs $a_i \in \Z^2, c_i \in \N^2$ (for $i = 0, \ldots, k$) and  $b_i \in \Z^2, d_i \in \N^2$ (for $i=1,\ldots, k$),
it can be checked in \pspace if they form the profile of some linear path scheme.
\end{lemma}
\begin{proof}
Guess intermediate control states $(q_1, p_1)$, \ldots, $(q_{k}, p_{k})$ and
put 
$(q_{k+1}, p_{k+1}) = (q, p)$.
Check that the following reachability properties hold in $\constrvass{\V}{\N, \N}$, for $i = 0, \ldots, k$ and $i = 1, \ldots, k$, respecively:
\begin{align*}
& \big( (q_{i}, p_{i}), c_i \big) \trans{} \big((q_{i+1}, p_{i+1}), c_i + a_i \big) \\ 
& \big( (q_{i}, p_{i}), d_i \big) \trans{} \big((q_i, p_i), d_i + b_i \big), 
\end{align*}
and that the above properties fail to hold if any $c_i$ (resp.~$d_i$) is replaced by a point-wise smaller pair of numbers.
All the required checks are instances of the reachability problem for 2-\vass, hence doable in \pspace~\cite{DBLP:conf/lics/BlondinFGHM15}.
\end{proof}

Denote by $\reach_p$ the set of configurations reachable in $\constrvass{\V}{\N, \N}$ via some linear path scheme
with profile $p$.
Using Lemma~\ref{lem:checknumbers} we can enumerate
all profiles of linear path schemes~\eqref{eq:lps} of length at most $N_1$ with $k \leq N_2$ loops.
Note that each such profile can be represented (in binary) in polynomial space:
as $N_1$ is exponential and $N_2$ is polynomial, the profile is defined by polynomially many numbers
$a_i, b_i, c_i, d_i$, each of them at most exponential.
Thus by the virtue of Lemma~\ref{lem:blondyn} it is enough to show, for a fixed profile $p$,
that the set $\reach_p$ is \pspacesolvable.
Fix a profile $p$ from now on.

As a convenient tool we will use \emph{linear Diophantine equations}.
These are systems of equations of the form
\begin{align} \label{eq:lineq}
a_1 x_1 + \cdots + a_l x_l = a,
\end{align}
where $x_1, \ldots, x_l$ are variables, and $a, a_1, \ldots, a_l$ are integer coefficients.
For a system $\U$ of such equations, 
we denote by $\sol{\U} \subseteq \N^l$ the solution set of $\U$,
i.e., the set all of non-negative integer vectors
$(n_1, \ldots, n_l)$ such that the valuation $x_1 \mapsto n_1, \ldots, x_l \mapsto n_l$ satisfies all the equations in $\U$.

We say that a vector is \emph{bounded} by $m$ if it is smaller than $m$ on every coordinate.
From the results of~\cite{D91,P91} (see also Prop.~2 in~\cite{taming} for the convenient formulation)
we get a bound on the size of semi-linear representation of the set of solutions of $\U$
which is polynomial in the number of variables and values of coefficients of $\U$,
but exponential in the number of equations:
%
%
\begin{lemmaC}[\cite{D91,taming,P91}]%
\label{lem:hybrid}
Let $\U$ be a system of $d$ linear Diophantine equations with $n$ variables such that the absolute values of
coefficients are bounded by $M$. Then
$\sol{\U} = B + P^*$, with every base $b\in B$ and period $p\in P$ bounded by
$\mathcal{O}{(n \cdot M)}^d$.
\end{lemmaC}
Observe that, forcedly, $P \subseteq \sol{\U_0}$ where
$\U_0$ denotes a modification of the system of linear equations $\U$ with all
right-hand side constants $a$ (cf.~\eqref{eq:lineq}) replaced by 0.
We will use Lemma~\ref{lem:hybrid} once we state the last lemma we need
(the idea is based on~\cite{DBLP:conf/lics/BlondinFGHM15}):
\begin{restatable}{lemma}{LemProfile}\label{lem:profile}
The set $\reach_p$ is a projection of the union
\[\sol{\U^1} \cup \ldots \cup \sol{\U^l},\]
for systems of linear Diophantine equations $\U^1 \ldots \U^l$
that can be enumerated in polynomial space.
\end{restatable}
\begin{proof}
Recall that $c_0 \in \N^2$ denotes the (point-wise) minimal nonnegative values of counters that allow to execute
$\alpha_0$ in $\constrvass{\V}{\N, \N}$.
We assume that $c_0 = (0,0)$, as otherwise the set $\reach_p$ is empty.
Introduce variables $x_i$, for  $i = 1 \ldots k$, to represent the number of times the loop $\beta_i$ has been executed.
The necessary and sufficient condition for executing the linear path scheme~\eqref{eq:lps}
can be described by a positive Boolean combination of linear inequalities (that can be further transformed into  linear
equations using auxiliary variables), which implies Lemma~\ref{lem:profile}.

Indeed, for every $i = 1,\ldots, k$, we distinguish two sub-cases, $x_i = 0$ or $x_i > 0$. In the former case
(the loop $\beta_i$ is \emph{not} executed) we put the two inequalities described succinctly as (note that $a_j, b_j \in \Z^2$, while $c_j$, $d_j \in \N^2$ for all $j$)
\begin{align*}
& a_0 + b_1 x_1 + a_1 + \cdots + b_{i-1} x_{i-1} + a_{i-1} \geq c_i
\end{align*}
to say that $\alpha_i$ can be executed.
In the latter case (the loop $\beta_i$ \emph{is} executed) we put the following four inequalities
\begin{align*}
& a_0 + b_1 x_1 + a_1 + \cdots + b_{i-1} x_{i-1} + a_{i-1}  \geq d_i  \\
& a_0 + b_1 x_1 + a_1 + \cdots + b_{i-1} x_{i-1} + a_{i-1} + b_i (x_i - 1)  \geq d_i
\end{align*}
to say that the first and the last iteration of the loop $\beta_i$ can be executed (which implies that each intermediate iteration of $\beta_i$ can be executed as well), plus the two inequalities
\begin{align*}
& a_0 + b_1 x_1 + a_1 + \cdots + b_{i-1} x_{i-1} + a_{i-1} + b_i x_i  \geq c_i
\end{align*}
to assure that $\alpha_i$ can be executed next.
Finally, in both cases the two variables $(y_1, y_2)$ representing the final configuration,
are related to other variables by the two equations:
\begin{align*}
& a_0 + b_1 x_1 + a_1 + \cdots + b_{k} x_{k} + a_{k}   =  (y_1, y_2).
\end{align*}
Thus for every sequence of  choices between  $x_i = 0$ and $x_i > 0$, for $i = 1, \ldots, k$, we obtain
a system of equations.
The set $\reach_p$ is the projection, onto $(y_1, y_2)$, of the union of solution sets of all these systems of equations.
\end{proof}

The last two lemmas immediately imply that $\reach_p$ is \pspacesolvable.
Indeed, by Lemma~\ref{lem:hybrid} applied to every of the systems $\U^i$, we have
$\sol{\U^i} = B_i+{P_i}^*$ for bases $B_i$ containing all vectors $b\in \sol{\U^i}$ bounded by $N$, and
periods $P_i$ containing all vectors $p\in \sol{{\U^i}_0}$ bounded by $N$, where $N$ is exponential and computable.
Relying on Lemma~\ref{lem:profile},
the algorithm enumerates all systems $\U^i$, then enumerates
all $b\in B_i$ satisfying the above constraints, and for each $b$ it enumerates all periods $p\in P_i$ satisfying the above constraints.
The proof of Proposition~\ref{prop:effsolv} is thus completed.

\subsection{Proof of Proposition~\ref{prop:effsolv2}}%
\label{sec:proofeffsolv2}
In the sequel we fix states $q, q'$ of $\A$ and $p, p'$ of $\B$, respectively.
Our aim is to prove that $\midd_{q q'}^{p p'}$ is \pspacesolvable, by encoding this set as Parikh image of an \ocn.

Recall that Parikh image $\pim{w}$ of a word $w \in \Sigma^*$, for a fixed ordering $a_1 < \cdots < a_k$
of $\Sigma$, is defined as the vector
$(n_1, \ldots, n_k)$ where $n_i$ is the number of occurrences of $a_i$ in $w$, for $i = 1, \ldots, k$. Parikh image lifts to languages: $\pim{L} = \setof{\pim{w}}{w\in L}$.

An \ocn we call \emph{1-\ocn} if all its transitions $(q, a, q', z)$ satisfy $z \in \set{-1, 0, 1}$.
We define a 1-\ocn $\C$ of exponential size, over a 5-letter alphabet
$\set{a_0, b_0, a_+, a_-, b_f}$, such that
$\midd_{q q'}^{p p'} $ is the image of a linear function of $\pim{L(\C)}$.
$\C$ starts with the zero counter value, and its execution splits into three phases.
In the first phase $\C$ reads arbitrarily many times $a_0$ without modifying the counter,
and arbitrary many times $b_0$, increasing the counter by $1$ at every $b_0$.
Thus the counter value of $\C$ at the end of the first
phase is equal to the number of $b_0$s.

In the last phase, $\C$ reads arbitrarily many times $b_f$,  decreasing the counter by $1$ at every $b_f$.
The accepting configuration of $\C$ requires the counter to be $0$.
Thus the counter value of $\C$ at the beginning of the last
phase must be equal to the number of $b_f$s.

In the intermediate phase $\C$ simulates the execution of $\constrvass{\V}{\Z, \N}$.
The counter value of $\C$ corresponds, during this phase, to the counter value of $\B$.
On the other hand, the counter value of $\A$ will only be reflected by the number of $a_+$ and $a_-$ read by $\C$.
States of $\C$ correspond to pairs of states of $\A$ and $\B$, respectively; there will be also exponentially many
auxiliary states.
The phase starts in state $(q, p)$, and ends in state $(q', p')$.
A transition $\big(  (q_1, p_1), (z_1, z_2), (q_2, p_2) \big)$ of $\V$ is simulated in $\C$ as follows:
First, if $z_1 \geq 0$ then $\C$ reads $z_1$ letters $a_+$; otherwise, $\C$ reads $-z_1$ letters $a_-$.
Second, if $z_2 \geq 0$ then $\C$ performs $z_2$ consecutive increments of the counter; otherwise $\C$ performs $-z_2$ decrements.
In both tasks, fresh auxiliary states are used.
We assume w.l.o.g.~that every transition of $\V$ satisfies $(z_1, z_2) \neq (0, 0)$; hence $\C$ has no $\varepsilon$-transitions.
This completes the description of the 1-\ocn $\C$.

Let $S = \pim{L(\C)} \subseteq \N^5$. Then
$\midd_{q q'}^{p p'} = f(S)$, for the linear function $f: \Z^5 \to \Z^4$ defined by
(intensionally, we re-use alphabet letters in the role of  variable names):
\[
(a_0, b_0, a_+, a_-, b_f) \mapsto (a_0, b_0, a_0 + a_+ - a_-, b_f).
\]
Therefore if $S$ is \pspacesolvable then $f(S)$ is also so; it thus remains to prove that $S$ is \pspacesolvable.

Our proof builds on results of~\cite{Hofman16,Kopczynski10}. In order to state it we need to introduce the concept of \emph{pump}
of an accepting run $\rho$ of $\C$ (called \emph{direction} in~\cite{Hofman16}). We treat accepting runs $\rho$ as sequences of transitions.
A pump of $\rho$ of the first kind is a sequence $\alpha$ of transitions such that $\rho$ factorizes into $\rho = \rho_1 \rho_2$, and
$\rho_1 \alpha \rho_2$ is again an accepting run. Note that in this case the effect of $\alpha$ on the counter is necessarily 0.
A pump of second kind is a pair $\alpha, \beta$ of sequences of transitions, where the effect of $\alpha$ is non-negative,
such that $\rho$ factorizes into
$\rho = \rho_1 \rho_2 \rho_3$, and $\rho_1 \alpha \rho_2 \beta \rho_3$ is again an accepting run. Note that
in this case the effect of $\beta$ is necessarily opposite to the effect of $\alpha$.

Parikh image of a sequence of transitions $\pim{\rho}$ is understood as a shorthand for Parikh image of the input word of $\rho$.
Furthermore, we use a shorthand notation for Parikh image of a pump $\pi$: let $\pim{\pi}$ mean either
$\pim{\alpha}$ or $\pim{\alpha \beta}$, in case of the first or second kind, respectively.
Similarly, the length of $\pi$ is either the length of $\alpha$, or the length of $\alpha \beta$.
Lemma~\ref{lem:hofman} is a direct consequence of Lemma 55 in~\cite{Hofman16arxiv}
(see also Lemma 15 in~\cite{Hofman16,Hofman16arxiv}, or Theorem 6 in~\cite{Kopczynski10}):
%
%
\begin{lemma}%
\label{lem:hofman}
There is a computable bound $N$, polynomial in $\size{\C}$, such that
$S$ is a union of linear sets of the form
\[
\pim{\rho} + \set{\pim{\pi_1}, \ldots, \pim{\pi_l}}^* \quad (l\leq 5),
\]
where $\rho$ is an accepting run of $\C$ of length at most $N$, and $\pi_1 \ldots \pi_l$ are
pumps of $\rho$ of length at most $N$.
\end{lemma}
\noindent
Except for different (but equivalent) notation, Lemma 55 in~\cite{Hofman16arxiv} differs from Lemma~\ref{lem:hofman}
only by claiming that a run $\rho$ determines uniquely pumps $\pi_1, \ldots, \pi_l$ (which is not needed here).

\newcommand{\logarithmic}{$\mathcal{O}(\log \size{C})$\xspace}
We need one more fact, in which we refer to the computable bound $N$ of Lemma~\ref{lem:hofman}
(note that \logarithmic is polynomial in the sizes of $\A$ and $\B$, as $\C$ is blown up exponentially
compared to $\A$ and $\B$):  
\begin{restatable}{lemma}{Lemsz}%
\label{lem:6}
For given $b\in \N^5$ and $P = \set{p_1, \ldots, p_l} \subseteq \N^5$, $l \leq 5$,
it is decidable in space \logarithmic
if there is an accepting run $\rho$ of $\C$ of length at most $N$
and pumps $\pi_1, \ldots, \pi_l$ of $\rho$ of length at most $N$,
such that $b = \pim{\rho}$ and $p_i = \pim{\pi_i}$ for $i = 1, \ldots, l$.
\end{restatable}
\begin{proof}
As the first step, the algorithm guesses, for each $i = 1, \ldots, l$, whether pump $\pi_i$ would be
of first or second kind.
For simplicity of presentation we assume below that all pumps are guessed to be of second kind, i.e.,
they are pairs $\alpha_i, \beta_i$
-- pumps of first kind are treated in a simpler but similar way.
Expanding the definitions, the algorithm is to check if there is an accepting run $\rho$ and pumps
$\alpha_i, \beta_i$ ($i = 1, \ldots, l$), all of length at most $N$ and of desired Parikh image,
so that for every $i$ the run factors into
\begin{align} \label{eq:factor5}
\rho \ = \ \rho^1_i \, \rho^2_i \, \rho^3_i
\end{align}
and the following sequence is again an accepting run:
\[
\rho_i \ = \ \rho^1_i \, \alpha_i \, \rho^2_i \, \beta_i \, \rho^3_i.
\]
To this aim the algorithm simulates in parallel, using space \logarithmic, at most $l+1 \leq 6$
different runs of $\C$, one of them corresponding to $\rho$ and the remaining $l$ ones
corresponding to $\rho_1, \ldots, \rho_l$.

Since $N$ is polynomial in $\size{\C}$, counting up to $N$ is possible in space \logarithmic.
The algorithm starts by simulating nondeterministically a run $\rho$ of $\C$ of length at most $N$.
We call this simulation \emph{main thread}
(in addition, there will be also at most $l$ auxiliary \emph{pump threads}, as introduced below).
The algorithm thus maintains the current configuration $c$ of $\C$ in the main thread,
and chooses nondeterministically consecutive transitions of $\C$ to execute.
In addition, the algorithm maintains Parikh image $a\in\N^l$
of the run executed so far, updated after every step of simulation and stored in space \logarithmic.
In the course of simulation the algorithm guesses nondeterministically $l$ positions
corresponding to the ends of prefixes $\rho^1_i$ in~\eqref{eq:factor5}, for $i = 1, \ldots, l$.
At each so guessed position all the threads are suspended, and a new \emph{pump} thread is added.
The new thread is responsible for simulating, from the current configuration $c$ of $\C$,
some sequence of transitions $\alpha_i$ of length at most $N$.
The simulation of $\alpha_i$ finishes nondeterministically, say in configuration $c_i$,
if the counter value of $c_i$ is greater or equal to the counter value of $c$ (the effect of $\alpha_i$ is non-negative)
and the control state of $c_i$ is the same as the control state of $c$.
Then Parikh image $a_i = \pim{\alpha_i}$ is stored (in space \logarithmic), and the simulation of
all threads (the suspended ones and the new one) are continued, with the proviso that all threads use \emph{the same}
nondeterministically chosen sequence of transitions.
Up to now, the algorithm maintains up to $l+1$ configurations
$c, c_1, \ldots, c_l$ of $\C$, and stores up to $l+1$ vectors $a, a_1 \ldots a_l \in \N^l$.

Later in the course of simulation the algorithm also guesses nondeterministically $l$ positions in $\rho$,
corresponding to the the beginnings of suffixes $\rho^3_i$ in~\eqref{eq:factor5}, for $i = 1, \ldots, l$.
Similarly as above, at each so guessed position all threads are suspended
except for the one corresponding to pump $i$, and that pump thread simulates some sequence of transitions $\beta_i$
of length at most $N$. This simulation terminates only if its current configuration becomes equal to the current
configuration of the main thread, i.e., $c_i = c$, and moreover Parikh image of $\beta_i$, say $b_i$, satisfies
$a_i + b_i = p_i$; and once this happens, the pump thread corresponding to pump $i$ is cancelled.
Note that the pump threads are not necessarily synchronized, and it might happen that
one pump thread becomes cancelled even before some another pump thread even starts.
This lack of synchronization is clearly not an issue, as the total number of simultaneously executed
pump threads stays below $l$.
The whole simulation finishes when all pump threads are cancelled, the current configuration $c$ is
the final configuration of $\C$, and Parikh image of the run executed in the main thread equals $b$.
\end{proof}

The last two lemmas imply that $S$ is \pspacesolvable. Indeed, it is enough to enumerate all
candidates $b, P$ bounded by $N$, as specified in Lemma~\ref{lem:hofman},
and validate them, using Lemma~\ref{lem:6}.
This completes the proof of Proposition~\ref{prop:effsolv2}.

%


\subsection{Proof of Proposition~\ref{prop:Reffsolv}}
Recall the definition~\eqref{eq:defR} of the set $\RR$ as a projection of a conjunction of constraints:
\begin{align*} \begin{aligned}
\defR
\end{aligned} \end{align*}
We are going to prove that $\RR$ is \pspacesolvable.
The algorithm enumerates quadruples of states $q, q', p, p'$. For each fixed quadruple,
it enumerates (by Propositions~\ref{prop:effsolv} and~\ref{prop:effsolv2})
component linear sets of $\pref_q^p$, $\midd_{q q'}^{p p'}$ and $\suff_q^p$.
Thus it is enough to consider three fixed  \pspacesolvable linear sets
\begin{align*}
L_\pref & =  b_\pref + {P_\pref}^*  \subseteq \pref_q^p \subseteq \N^2 \\
L_\midd & =  b_\midd + {P_\midd}^*  \subseteq \midd_{q q'}^{p p'} \subseteq \N^2\times\Z\times\N \\
L_\suff & =  b_\suff + {P_\suff}^*  \subseteq \suff_q^p \subseteq \N^2 .
\end{align*}

We now treat each of these linear sets as a system of linear Diophantine equations.
For instance, if $P_\pref = \set{p_1, \ldots, p_k}$, all pairs $(x, y) \in L_\pref$ are described by two equations
\begin{align} \label{eq:coeff}
(x_\pref, y_\pref) = b_\pref + p_1 x_1 + \cdots + p_k x_k,
\end{align}
for fresh variables $x_1, \ldots, x_k$. Note that the number of variables is exponential, but the number of equations
is constant, namely equal two.
The same can be done with two other linear sets $L_\midd$ and $L_\suff$,
yielding 6 more equations involving 6 variables
$x_\midd, y_\midd, x'_\midd, y'_\midd, x_\suff, y_\suff$, plus exponentially many other fresh variables.
Points in $L_\suff$ are represented as $(x_\suff, y_\suff)$, while points in $L_\midd$ as $(x_\midd, y_\midd, x'_\midd, y'_\midd)$.
(Since we only consider non-negative solutions of systems of equations, in case of
$L_\midd$ two separate cases should be considered, depending on whether the final value $x'_\midd$
is non-negative or not. For simplicity we only consider here the case when it is non-negative.)
In addition, we add the following equations (and one variable $x$):
\begin{align*}
x_\pref & = x_\midd \\
y_\pref & = y_\midd \\
y'_\midd & = y_\suff \\
x & = x'_\midd - x_\suff.
\end{align*}
In total, we have a system $\U$ of 12 equations (8 mentioned before and 4 presented above) involving exponentially many variables,
including these 9 variables
\begin{align} \label{eq:vars}
x_\pref, y_\pref, x_\midd, y_\midd, x'_\midd, y'_\midd, x_\suff, y_\suff, x.
\end{align}
The value of 3 of them is relevant for us, namely $x_\pref, x_\suff$ and $x$.
We claim that the projection of the solution set $\sol{\U}$ onto these 3 \emph{relevant} coordinates is \pspacesolvable.
To prove this we apply once more Lemma~\ref{lem:hybrid},
%
%
to deduce that
the solution set of $\U$ is of the form $B + P^*$, for all bases $b \in B$ and all periods $p\in P$ bounded
exponentially.
Hence each number appearing in a base or a period is representable in polynomial space; the same applies to the projection
onto the 3 relevant coordinates.
Finally, observe that the projections of the sets $B$ and $P$ onto the 3 relevant coordinates
can be enumerated in polynomial space.
Indeed, in order to check whether a vector in $\N^3$ is (the projection of) a $B$ or $P$, we proceed similarly as in the proof of Proposition~\ref{prop:effsolv}:
the algorithm first guesses values of other 6 variables among~\eqref{eq:vars},
and then guesses the values of other (exponentially many) variables on-line,
in the course of enumerating the coefficients $p_i$ (cf.~\eqref{eq:coeff}); the latter is possible as the
sets $L_\pref, L_\midd$ and $L_\suff$ are \pspacesolvable.



\section{\pspace lower bound}%
\label{sec:pspace-hard}

Recall that a language is \emph{\good}  if it is a finite Boolean combination of languages of the form $w\Sigma^*$, for $w \in \Sigma^*$.
In this section we prove the following result which, in particular, implies the lower bound of Theorem~\ref{thm:pspace-comp}:

\begin{theorem}\label{thm:pspacehard}
For every class $\F$ containing all \good languages,
the $\F${\sepsep}separability problem for languages of \ocn  is \pspace-hard.
\end{theorem}

A convenient \pspace-hard problem, to be reduced to $\F$ separability of \ocn, can be extracted from~\cite{FJ15}.
Given an \oca $\A$ and $b\in \N$ presented in binary, the \emph{bounded} non-emptiness problem asks whether $\A$ accepts some word
by a \emph{$b$-bounded} run; a run is $b$-bounded if counter values along the run are at most $b$.

\begin{theoremC}[\cite{FJ15}]
The bounded non-emptiness problem is \pspace-complete, for $\A$ and $b$ represented in binary.
\end{theoremC}

A detailed analysis of the proof reveals that the problem remains \pspace-hard even if the input \oca
$\A = (Q, \alpha_0, \alpha_f, T, T_{=0})$ is assumed to be \emph{acyclic},
in the sense that there is no reachable configuration $\alpha$ 
with a non-empty
run
$\alpha \trans{} \alpha$. Observe that an acyclic \oca has no $b$-bounded run longer than $b |Q|$, a property which
will be crucial for the correctness of our reduction.

\begin{proposition}%
\label{prop:acyclic}
The bounded non-emptiness problem is \pspace-complete, for acyclic $\A$ and $b$ represented in binary.
\end{proposition}

We are now ready to prove Theorem~\ref{thm:pspacehard}, by reduction from bounded non-emptiness of acyclic \oca.
Given an acyclic \oca $\A= (Q, (q_0, 0), (q_f, 0), T, T_{=0})$ and $b\in\N$, we construct in polynomial time two \ocn
$\B$ and $\B'$, with the following properties:

\begin{enumerate}
\item[(a)] if $\A$ has a $b$-bounded accepting run then $L(\B) \cap L(\B')\neq \emptyset$
(and thus $L(\B)$ and $L(\B')$ are not $\F$ separable); 
\item[(b)] if $\A$ has no $b$-bounded accepting run then $L(\B)$ and $L(\B')$ are $\F$ separable. 
\end{enumerate}

The two \ocn $\B$ and $\B'$ will jointly simulate a $b$-bounded run of $\A$, obeying an invariant that
the counter value $v$ of $\B$ is the same as the counter value of $\A$, while the counter value of $\B'$ is $b-v$.
The actual input alphabet of $\A$ is irrelevant; as the input alphabet of $\B$ and $\B'$ we take $\Sigma = T \cup T_{=0}$.
The \ocn $\B$ behaves essentially as $\A$, except that it always allows for a zero test.
Formally, $\B= (Q, (q_0, 0), (q_f, 0), U)$, where the transitions $U$ are defined as follows.
For every transition $t = (q, a, q', z) \in T$, there is a corresponding transition
$
(q, t, q', z) \in U.
$
Moreover, for every zero test $t = (q, a, q') \in T_{=0}$, there is a transition
$
(q, t, q', 0) \in U.
$
On the other hand, the \ocn $\B'$ starts in the configuration $(q_0, b)$, ends in $(q_f, b)$,
and simulates the transitions of $\A$ but with the opposite effect.
Formally, $\B' = (Q \cup X, (q_0, b), (q_f, b), U')$, for a set $X$ of auxiliary states.
For every transition $t = (q, a, q', z) \in T$, there is a corresponding transition
$
(q, t, q', -z) \in U
$
with the effect $-z$ opposite to the effect of $t$.
Moreover, for every zero test $t = (q, a, q') \in T_{=0}$, we include into $U'$ the following three transitions
\[
(q, \varepsilon, p, -b) \qquad
(p, \varepsilon, p', +b)  \qquad
(p', t, q', 0),
\]
for some auxiliary states $p, p'$.
The aim of the first two transitions is to allow the last one only if the counter value is
at least $b$ (and thus exactly $b$, assuming there is also a run of $\B$ on the same input).

We need to argue that the implications~(a) and~(b) hold.
The first one is immediate: every $b$-bounded accepting run of $\A$ is faithfully simulated by $\B$ and $\B'$, and thus
the languages $L(\B)$ and $L(\B')$ have non-empty intersection.

For the implication~(b), suppose $\A$ has no $b$-bounded accepting run.
The first step is to notice that the languages $L(\B)$ and $L(\B')$ are necessarily disjoint.
Indeed, any word $w \in L(\B) \cap L(\B')$ would describe a $b$-bounded accepting run of $\A$:
$\B$ ensures that the counter remains non-negative, while $\B'$ ensures that the counter does not increase beyond $b$
and that the zero tests are performed correctly.

Let $L$ contain all prefixes of words from $L(\B)$, and likewise $L'$ for $L(\B')$.
Let $n = b |Q|$. Recall that due to acyclicity, $\A$ has no $b$-bounded run of length $n$ (in the sense of the number of transitions)
or longer.
Thus, for the same reason as above, the intersection $L \cap L'$ contains no word of length $n$ or longer.

In simple words, we are going to show that
for every word $w \in L(\B) \cup L(\B')$,
looking at the prefix of length $n$ of $w$ suffices to decide whether
$w\in L(\B)$ or $w\in L(\B')$.

We define a language $K\in\F$ as follows:
\[
K  \ := \ \big(L(\B) \cap \Sigma^{< n}\big) \ \cup \ \bigcup_{w\in L, |w| = n} w \Sigma^*,
\]
where $\Sigma^{<n}$ stands for the set of all words over $\Sigma$ of length strictly smaller than $n$, and
$|w|$ denotes the length of $w$.
The language $K$ belongs to $\F$ indeed, as $\F$ is closed under finite unions, and
every singleton $\{w\}$ belongs to $\F$, due to
\[
\{w\} = w\Sigma^* - \bigcup_{a \in \Sigma} wa\Sigma^*.
\]
It remains to argue that $K$ separates $L(\B)$ and $L(\B')$. By the very definition $L(\B) \subseteq K$,
as $K$ contains all words from $L(\B)$ of length strictly smaller than $n$,
and all words starting with a prefix,  of length $n$, of a word from $L(\B)$.
For disjointness of $K$ and $L(\B')$, observe that the languages $L(\B) \cap \Sigma^{< n}$ and
$L(\B')$ are disjoint, as already $L(\B)$ and $L(\B')$ are.
Moreover, for every $w\in L$ of length $|w| = n$, the languages $w\Sigma^*$ and $L(\B')$ are disjoint, as already
the intersection $L\cap L'$ contains no word of length $n$ or longer.   

\begin{remark} \rm
The \ocn $\B$ and $\B'$ used in the reduction can be easily made deterministic.
On the other hand, by a general result of~\cite{CCLP16} we learn that regular separability of nondeterministic \ocn polynomially reduces to
regular separability of \emph{deterministic} \ocn, making the latter \pspace-complete too.
\end{remark}



\section{Undecidability for one counter automata}%
\label{sec:undecid}

In this section we prove Theorem~\ref{thm:undecidability}.
The argument is similar to the proof of the previous section, except that instead of reducing a fixed undecidable problem,
we provide a polynomial reduction from \emph{every} decidable one.
This idea derives from the insight of~\cite{DBLP:journals/jacm/Hunt82a}.
%

A universal model of computation that will be convenient for us is 2-counter machines.
A \emph{deterministic 2-counter machine} $\M$ consists of a finite set of \emph{states} $Q$ with distinguished \emph{initial} state $q_0 \in Q$,
\emph{accepting} state $q_\acc \in Q$ and \emph{rejecting} state $q_\rej \in Q$,
two \emph{counters} $c_1, c_2$,
and a set of transitions, one per state $q\in Q - \set{q_\acc, q_\rej}$.
Thus the accepting state and the rejecting one have no outgoing transitions.
There are two types of transitions. Type 1 transitions increment one of the counters ($i \in \set{1, 2}$) by one,
and type 2 transitions conditionally decrement one of the counters by one:
\newcommand{\transone}[1]{{\small \bf (type 1) in state $q$, increment $c_{#1}$ and go to state $q'$}}
\newcommand{\transtwo}[1]{{\small \bf (type 2) in state $q$, if $c_{#1}>0$ then decrement $c_{#1}$ and go to state $q'$, else go to state $q''$}}
\begin{enumerate}
\item[] \transone{i};
\item[] \transtwo{i}.
\end{enumerate}
A configuration $(q, n_1, n_2)$ of $\M$ consists of a state $q$ and values $n_1, n_2 \geq 0$ of the counters.
We write $(q, n_1, n_2) \trans{} (q', n'_1, n'_2)$ if a sequence of transitions leads from configuration $(q, n_1, n_2)$  to
 $(q', n'_1, n'_2)$.
%
%
We say that $\M$ \emph{accepts} a number $k \in \N$ if $(q_0, k, 0) \trans{} (q_\acc, 0, 0)$,
and \emph{rejects} $k$ if $(q_0, k, 0) \trans{} (q_\rej, 0, 0)$.
Note our specific requirement that acceptance or rejection only happens with both counter values equal to 0.
The machine $\M$ is \emph{total} if every $k \in \N$ is either accepted or rejected by $\M$.
The language $L(\M)$ recognized by $\M$ is set of all numbers accepted by $M$.

Every decidable language, say over the alphabet $\set{0, 1}$, is recognized by some total, deterministic
2-counter machine, under a suitable encoding.
Indeed, every word $w \in \set{0, 1}^*$ can be encoded, using binary representation, as a natural number $n(w)$.
It is quite standard to show that then for every total deterministic Turing machine $\T$,
there is a total deterministic 2-counter machine $\M$ such that
$w \in L(\T)$ if, and only if, $2^{n(w)} \in L(\M)$.
\footnote{
The exponent arises from the standard simulation of a Turing machine by a 3-counter machine;
the latter is further simulated by a 2-counter machine which stores the values of the 3 counters $c, d, e$
in the form $2^c 3^d 5^e$.
}
Thus, modulo the encoding, decidable languages are a subclass of (in fact, the same class as)
subsets $L \subseteq \N$ of natural numbers recognized by total deterministic 2-counter machines.
These subsets $L\subseteq \N$ we call below \emph{decidable problems}.

Let $\F$ be a class of languages containing all \good languages.
We are going to show a polynomial time reduction from any decidable problem $L \subseteq \N$ to
$\F${\sepsep}separability of \oca languages.
This implies undecidability of the latter problem.
Indeed, suppose $\F${\sepsep}separability of \oca languages is decidable,
and let $f : \N \to \N$ be any space constructible function such that
$\F${\sepsep}separability of \oca is decidable in $f(n)$ space, where $n$ is the size of input.
As a consequence, every decidable problem $L \subseteq \N$ would be actually decidable in space
$f(p(n))$ for some polynomial $p$.
This would contradict the space hierarchy theorem (see, e.g., Thm.~9.3 in~\cite{Sipserbook}),
according to which there are problems decidable in space $\mathcal{O}(f(2^n))$ but not in space $o(f(2^n))$.
\begin{proposition}\label{prop:oca-reduction}
Every decidable problem $L \subseteq \N$ reduces polynomially
to the $\F$ separability problem of \oca languages.
\end{proposition}

\begin{proof}
Let $\M$ be a fixed total deterministic 2-counter machine recognizing a language $L$.
Given $k \in \N$, we construct two \oca $\A_1, \A_2$ with the following properties:
\begin{enumerate}
\item[(a)] if $k\in L(\M)$ then $L(\A_1) \cap L(\A_2) \neq \emptyset$
(and thus $L(\A_1)$ and $L(\A_2)$ are not $\F$ separable);  
\item[(b)] if $k \notin L(\M)$ then $L(\A_1)$ and $L(\A_2)$ are $\F$ separable.  
\end{enumerate}
As the input alphabet $\Sigma$ of $\A_1$ and $\A_2$ we take the set of transitions of $\M$.
We define two \oca:
\begin{align*}
\A_1 & = (Q, (q_0, k), (q_\acc, 0), T_1, T_{1, = 0}), \\
\A_2 & = (Q, (q_0, 0), (q_\acc, 0), T_2, T_{2, = 0}),
\end{align*}
where transitions $T_1$ (resp.~$T_2$) and zero tests $T_{1, =0}$ (resp.~$T_{2, =0}$) are, roughly speaking, transitions of
$\M$ where the second (resp.~first) counter is ignored.
Formally, for every transition $t$ of $M$ of type 1 on counter $c_1$ of the form
\begin{enumerate}
\item[] \transone{1},
\end{enumerate}
there is a transition
$
(q, t, q', +1) \in T_1;
$
and for every transition $t$ of $M$ of type 1 on counter $c_2$ of the form
\begin{enumerate}
\item[] \transone{2},
\end{enumerate}
there is a transition
$
(q, t, q', 0) \in T_1.
$
Furthermore, for every transition $t$ of $M$ of type 2 on counter $c_1$ of the form
\begin{enumerate}
\item[] \transtwo{1},
\end{enumerate}
we include the following transition and zero test:
\[
(q, t, q', -1) \in T_1 \qquad (q, t, q'') \in T_{1, =0}.
\]
Finally, for every transition $t$ of $M$ of type 2 on counter $c_2$ of the form
\begin{enumerate}
\item[] \transtwo{2},
\end{enumerate}
we include the following two transitions:
\[
(q, t, q', 0) \in T_1 \qquad (q, t, q'', 0) \in T_1.
\]
Transitions and zero tests of $\A_2$ are defined symmetrically, with the roles of $c_1$ and $c_2$ swapped.

We need to argue that the implications~(a) and~(b) hold.
The first one is immediate: every sequence of transitions of $\M$ leading from $(q_0, k, 0)$ to $(q_\acc, 0, 0)$,
treated as a word over $\Sigma$, belongs both to $L(\A_1)$ and $L(\A_2)$.

In order to prove implication~(b), suppose $k\notin L(\M)$. We first observe that $L(\A_1)$ and $L(\A_2)$ are
necessarily disjoint; indeed, any $w\in L(\A_1) \cap L(\A_2)$ is a sequence of transitions that accepts $k$.

As $\M$ is total by assumption, we know that $(q_0, k, 0) \trans{} (q_\rej, 0, 0)$ in $\M$;
let $n$ be the length of the corresponding sequence of transitions.

Let $L_1$ contain all prefixes of words from $L(\A_1)$, and likewise $L_2$ for $L(\A_2)$.
It is crucial to observe that the intersection $L_1 \cap L_2$ contains no word of length $n$ or longer.
Indeed, any $w\in L_1 \cap L_2$ is a sequence of transitions of $\M$ starting from $(q_0, k, 0)$,
and thus cannot be longer than $n$. Moreover $w \in L_1 \cap L_2$ cannot
lead, as a sequence of transitions of $\M$, to the rejecting state (as it has no outgoing transitions),
and thus $w$ can not have length $n$ either.

The rest of the proof is along the same lines as in the previous section.
In simple words, we claim that for a word of length $n$ or longer, it is enough to inspect its prefix of length $n$ in order
to classify the word between $L(\A_1)$ and $L(\A_2)$.
Formally,  we define a language $K\in\F$ as follows:
\[
K \ := \ \big(L(\A_1) \cap \Sigma^{< n}\big) \ \cup \ \bigcup_{w\in L_1, |w| = n} w \Sigma^* .
\]
The language $K$ belongs to $\F$ for the reasons discussed in the previous section.
It remains to argue that $K$ separates $L(\A_1)$ and $L(\A_2)$.
By the very definition $L(\A_1) \subseteq K$,
as $K$ contains all words from $L(\A_1)$ of length strictly smaller than $n$,
and all words starting with a prefix, of length $n$, of a word from $L(\A_1)$.
For disjointness of $K$ and $L(\A_2)$, observe that the languages $L(\A_1) \cap \Sigma^{< n}$ and $L(A_2)$ are disjoint,
as already $L(\A_1)$ and $L(\A_2)$ are.
Moreover, for every $w\in L_1$ of length $|w| = n$, the languages $w\Sigma^*$ and $L(\A_2)$ are disjoint, as already
the intersection $L_1\cap L_2$ contains no word of length $n$ or longer.
\end{proof}

\begin{remark} \rm
Theorem~\ref{thm:undecidability} is used in~\cite{CCLP16} to prove undecidability of the regular separability problem
for visibly one counter automata (cf.~Theorem 5 in~\cite{CCLP16}).
The proof assumes that the alphabets of the input visibly one counter automata can be different;
on the other hand, when two input visibly one counter automata are assumed to be over the same alphabet (i.e., they perform
increment and decrement operations on the same input letters) the problem seems to be decidable~\cite{ChristofsDecidability}.
This shows that the decidability border is quite subtle. In addition,
the regular separability problem becomes once more undecidable when one
extends visibly one-counter automata over the same alphabet to
visibly pushdown automata over the same alphabet, as shown in~\cite{Kopczynski16}.
\end{remark}



\section{Final remarks}%
\label{sec:remarks}

Our main contribution is to show that  the regular separability problem for \ocn
is decidable (we also provide tight complexity estimation of the problem, namely \pspace-completeness, which we consider however less significant),
but it becomes undecidable for \oca (when zero tests are allowed).
We believe that this reveals a delicate decidability borderline.
For instance recall (cf.~Remark~\ref{rem:oca}) that the concept of $n$-approximation, a core technical ingredient of our decidability proof,
still works for \oca, including the Approximation Lemma, but is not prone to effective testing.
Below we discuss in more detail two other aspects:
relation to the regularity problem for \ocn, and obstacles towards extending our approach to regular separability
of the many-dimensional extension of \ocn, i.e., of \vass.

\myparagraph{Undecidability of regularity}
Our decidability result contrasts with undecidability of
the regularity problem for \ocn (given an \ocn $\A$, decide if $L(\A)$ is regular?), shown in~\cite{ValkVidal81}.
The proof of~\cite{ValkVidal81} works for \ocn accepting by final configuration
(as assumed in this paper, cf.~Section~\ref{sec:oca}), but not for \ocn accepting solely by final state.
But even in this weaker model the regularity problem is undecidable, as discovered recently by
James Worrell~\cite{BensUndecidability}.
The proof is by reduction from finiteness of the reachability set of a lossy counter machine,
which is an undecidable problem~\cite{cheatsheet10}.
Consider a standard encoding of runs of such a machine as words, and consider the language of
\emph{reverses} of such encodings, i.e., encodings read backward.
It is not difficult to prove that the language is regular if, and only if, the reachability set of the lossy counter machine is finite.
Moreover, one can construct an \ocn that recognizes the complement of the language.


\myparagraph{Towards regular separability of \vass}
Our decidability proof builds upon a notion of $n$-approximation: an \ocn $\A$ is over-approximated
by an NFA $\A_n$ which remembers the counter value of $\A$ exactly only below $n$, and modulo $n$ above this threshold.
Could one define $n$-approximation $\V_n$ of a \vass $\V$ by treating all the counters of $\V$
in that way?
In particular, such $n$-approximation would commute with the \crossproduct:
$\V_n \crossprod \U_n = {(\V \crossprod \U)}_n$ for two \vass $\V$ and $\U$
(we extend here naturally the \crossproduct operation).

The Approximation Lemma (cf.~Lemma~\ref{lem:appr}), quite surprisingly, does not hold for so defined
notion of over-approximation.
Indeed, the Approximation Lemma would imply that regular separability of
$\V$ and $\U$ is equivalent to disjointness of languages of $\V_n$ and $\U_n$, for some $n>0$
(cf.~Corollary~\ref{cor:appr}), which is the same as
$L(\V_n \crossprod \U_n) = L({(\V \crossprod \U)}_n) = \emptyset$ for some $n > 0$;
and finally, the latter condition would be equivalent, again due to the
Approximation Lemma, to $L(\V \crossprod \U) = \emptyset$, which is the same
as the languages of $\V$ and $\U$ being disjoint.
Thus regular separability of $\V$ and $\U$ would be equivalent to disjointness of $\V$ and $\U$, which is not true in general.

The decidability status of the regular separability problem for \vasslangs remains thus open.

\myparagraph{\ocn versus \oca}
As remarked in Section~\ref{sec:approx},  the Approximation Lemma and its consequences
(Corollaries~\ref{cor:apprlemma} and~\ref{cor:appr}) can be shown for \oca just as well.
This observation seems to open the way to decidability of the regular separability problem
when one of input devices is \oca and the other one is \ocn. We conjecture the problem to be decidable,
and even to belong to \pspace, relying on the effective semi-linearity of the reachability relation
for \ocn with one zero test~\cite{FLS18}.



%
%


%


\section*{Acknowledgment}
We thank James Worrell for allowing us to include his undecidability proof of \ocn regularity~\cite{BensUndecidability}\footnote{
Found out at \emph{Autob{\'o}z'16}, the annual research camp of Warsaw automata group and friends.},
and Christoph Haase, Mohnish Pattathurajan,
Mahsa Shirmohammadi, Patrick Totzke and Georg Zetzsche for fruitful discussion and valuable suggestions.
We are also grateful to the reviewers for their valuable comments which helped to improve the presentation
significantly.

\bibliographystyle{plain}
\bibliography{citat}

\end{document}